\documentclass[journal=tosc,preprint]{iacrtrans}

\usepackage[utf8]{inputenc}
\usepackage[bookmarksdepth=3]{hyperref}
\usepackage{standalone}
\usepackage{xspace}
\usepackage{enumitem}
\usepackage{array}
\usepackage{booktabs}
\usepackage{blkarray}
\usepackage{cleveref}
\usepackage{orcidlink}

\usetikzlibrary{arrows.meta}

\setlist[enumerate,1]{label={(\arabic*)}}

\allowdisplaybreaks

\theoremstyle{plain}
\newtheorem{thm}{Theorem}[section]
\newtheorem{lem}[thm]{Lemma}
\newtheorem{cor}[thm]{Corollary}

\newtheorem{defn}[thm]{Definition}

\theoremstyle{definition}
\newtheorem{rem}[thm]{Remark}
\newtheorem{ex}[thm]{Example}
\theoremstyle{plain}

\crefname{thm}{theorem}{theorems}
\Crefname{thm}{Theorem}{Theorems}
\crefname{defn}{definition}{definitions}
\Crefname{defn}{Definition}{Definitions}
\crefname{prop}{proposition}{propositions}
\Crefname{prop}{Proposition}{Propositions}
\crefname{lem}{lemma}{lemmas}
\Crefname{lem}{Lemma}{Lemmas}
\crefname{cor}{corollary}{corollaries}
\Crefname{cor}{Corollary}{Corollaries}
\crefname{ex}{example}{examples}
\Crefname{ex}{Example}{Examples}
\crefname{rem}{remark}{remarks}
\Crefname{rem}{Remark}{Remarks}
\crefname{hypot}{hypothesis}{hypotheses}
\Crefname{hypot}{Hypothesis}{Hypotheses}
\crefname{conj}{conjecture}{conjectures}
\Crefname{conj}{Conjecture}{Conjectures}

\makeatletter
\newcommand{\subalign}[1]{%
    \vcenter{%
        \Let@ \restore@math@cr \default@tag
        \baselineskip\fontdimen10 \scriptfont\tw@
        \advance\baselineskip\fontdimen12 \scriptfont\tw@
        \lineskip\thr@@\fontdimen8 \scriptfont\thr@@
        \lineskiplimit\lineskip
        \ialign{\hfil$\m@th\scriptstyle##$&$\m@th\scriptstyle{}##$\hfil\crcr
            #1\crcr
        }%
    }%
}
\makeatother

\newcolumntype{M}[1]{>{\centering\arraybackslash}m{#1}}

    \newcommand{\F}{\mathbb{F}}                                         
    \newcommand{\Fq}{\F_{q}}                                            
    \newcommand{\Fqn}{\Fq^{n}}                                          
    \newcommand{\Fqx}{\Fq^{\times}}                                     
    \newcommand{\abs}[1]{\left\vert #1 \right\vert}                     
    \newcommand{\ceil}[1]{\left\lceil #1 \right\rceil}                  
    \newcommand{\floor}[1]{\left\lfloor #1 \right\rfloor}               
    \DeclareMathOperator{\rowspace}{rowsp}                              
    \newcommand{\Z}{\mathbb{Z}}                                         
    \newcommand{\homog}{\text{\normalfont hom}}                         
    \DeclareMathOperator{\solvdeg}{sd}                                  
    \DeclareMathOperator{\reg}{reg}                                     
    \newcommand{\topcomp}{\text{\normalfont top}}                       
    \DeclareMathOperator{\circulant}{circ}                              
    \DeclareMathOperator{\diag}{diag}                                   
    \DeclareMathOperator{\rank}{rank}                                   
    \newcommand{\degree}[1]{\deg \left( #1 \right)}                    
    \DeclareMathOperator{\LM}{LM}                                       
    \DeclareMathOperator{\gl}{GL}                                       
    \newcommand{\GL}[2]{\gl_{#1} \left( #2 \right)}                     

\newcommand{\Ciminion}{\texttt{Ciminion}\xspace}
\newcommand{\Ciminiontwo}{\texttt{Ciminion2}\xspace}
\newcommand{\Hydra}{\textsf{Hydra}\xspace}
\newcommand{\SageMath}{\texttt{SageMath}\xspace}
\newcommand{\OSCAR}{\texttt{OSCAR}\xspace}

\newcommand{\repository}{\url{https://github.com/sca-research/Groebner-Basis-Cryptanalysis-of-Ciminion-and-Hydra.git}}

\title{Gr\"obner Basis Cryptanalysis of Ciminion and Hydra}
\author{Matthias Johann Steiner \orcidlink{0000-0001-5206-6579}}
\authorrunning{M.~J.~Steiner}
\institute{Alpen-Adria-Universit\"at Klagenfurt, Klagenfurt am W\"orthersee, Austria \\ \email{mattsteiner@edu.aau.at}}

\begin{document}
    \maketitle

    \keywords{Gr\"obner basis \and Ciminion \and Hydra}
    \begin{abstract}
        \texttt{Ciminion} and \textsf{Hydra} are two recently introduced symmetric key Pseudo-Random Functions for Multi-Party Computation applications.
        For efficiency, both primitives utilize quadratic permutations at round level.
        Therefore, polynomial system solving-based attacks pose a serious threat to these primitives.
        For \texttt{Ciminion}, we construct a quadratic degree reverse lexicographic (DRL) Gr\"obner basis for the iterated polynomial model via linear transformations.
        With the Gr\"obner basis we can simplify cryptanalysis, as we no longer need to impose genericity assumptions to derive complexity estimates.
        For \textsf{Hydra}, with the help of a computer algebra program like \texttt{SageMath} we construct a DRL Gr\"obner basis for the iterated model via linear transformations and a linear change of coordinates.
        In the \textsf{Hydra} proposal it was claimed that $r_\mathcal{H} = 31$ rounds are sufficient to provide $128$ bits of security against Gr\"obner basis attacks for an ideal adversary with $\omega = 2$.
        However, via our \textsf{Hydra} Gr\"obner basis standard term order conversion to a lexicographic (LEX) Gr\"obner basis requires just $126$ bits with $\omega = 2$.
        Moreover, using a dedicated polynomial system solving technique up to $r_\mathcal{H} = 33$ rounds can be attacked below $128$ bits for an ideal adversary.
    \end{abstract}

    \section{Introduction}
    With secure \emph{Multi-Party Computation} (MPC) several parties can jointly compute a function without revealing their private inputs.
    In practice, data often has to be securely transferred from and to third parties before it can be used in MPC protocols.
    Also, one may need to store intermediate results securely in a database.
    For efficiency, it is recommended to perform encryption via an MPC-friendly \emph{Pseudo-Random Function} (PRF) \cite{CCS:GRRSS16} which utilizes a secret-shared symmetric key.
    Compared to classical bit-based PRFs their MPC counterparts have two novel design criteria:
    \begin{itemize}
        \item They must be native over large finite fields $\Fq$, typically $\log_2 \left( q \right) \geq 64$ and $q$ is prime.

        \item They must admit a low number of multiplications for evaluation.
    \end{itemize}
    Two such recently introduced MPC-friendly PRFs are \Ciminion \cite{EC:DGGK21} and \Hydra \cite{EC:GOSW23}.
    These constructions are based on the iteration of degree $2$ Feistel and Lai--Massey permutations.
    Due to the low degree at round level, both designs consider Gr\"obner basis attacks as one of the most threatening attack vectors.

    In this paper we perform dedicated Gr\"obner basis cryptanalysis for both designs.
    Let us quickly sketch the main steps of a Gr\"obner basis attack on a symmetric design:
    \begin{enumerate}
        \item Model the function with a system of polynomials.
        For \Ciminion and \Hydra we analyze an \emph{iterated} polynomial model $\mathcal{R} (\mathbf{x}_{i - 1}, \mathbf{y}) - \mathbf{x}_i = \mathbf{0}$, where $\mathcal{R}: \Fqn \to \Fqn$ is the round function, the $\mathbf{x}_i$'s are intermediate state variables with $\mathbf{x}_0$ the input state and $\mathbf{x}_r$ the output state of the function, and $\mathbf{y}$ the cryptographic secret variables.

        \item\label{Item: DRL step} Compute a Gr\"obner basis with respect to an efficient term order, e.g.\ the \emph{degree reverse lexicographic} (DRL) term order.

        \item\label{Item: LEX step} Construct a univariate polynomial, e.g.\ via term order conversion to the \emph{lexicographic} (LEX) term order or via the Eigenvalue Method.

        \item Factor the univariate polynomial.
    \end{enumerate}

    For \Ciminion and \Hydra the designers conjectured that Step \ref{Item: DRL step} dominates the complexity of a Gr\"obner basis attack.
    In addition, to estimate the complexity of this step it was assumed that the polynomial systems are \emph{regular} or \emph{semi-regular} \cite{Bardet-Complexity}.
    We on the other hand will construct DRL Gr\"obner bases for both designs via linear transformations and possibly a linear change of variables.
    Thus, we can directly proceed with Step \ref{Item: LEX step}.
    Typically, this step is also computationally cheaper than finding the DRL Gr\"obner basis for generic polynomial systems, which is the main reason why we yield better complexity estimations than the designers.

    \subsection{Our Results \& Relationship to Other Works}\label{Sec: Our Results}
    The research goals of this work can be summarized in three bullet points:
    \begin{itemize}
        \item Increase the understanding of polynomial models for MPC-friendly designs.

        \item Provide clean cryptanalysis which avoids strong genericity assumptions.

        \item Improve the complexity estimations of \Ciminion and \Hydra from a designer's point of view.
    \end{itemize}

    \paragraph{\Ciminion.}
    The PRF \Ciminion is an iterated Feistel design based on the Toffoli gate \cite{Toffoli-ReversibleComputing}.
    For the iterated \Ciminion polynomial model we construct a DRL Gr\"obner basis via linear transformations, see \Cref{Th: Ciminion Groebner basis}.
    Hence, we trivialize the hardness assumption of Step \ref{Item: DRL step}.

    It is worthwhile mentioning that in a recent work Bariant \cite{Bariant-Ciminion} constructed the univariate LEX polynomial of \Ciminion via iterated univariate polynomial multiplications.
    The complexities of constructing and factoring the LEX polynomial are linear-logarithmic in its degree.
    On the other hand, constructing the univariate LEX polynomial via Gr\"obner basis methods scales at least quadratic-logarithmic in its degree.
    Therefore, to the best of our knowledge Bariant's attack currently constitutes the most competitive attack on certain \Ciminion parameters.
    We also review this attack in more detail in \Cref{Sec: Bariant's attack}.

    There are two simple countermeasures against Bariant's attack: increasing the number of rounds or additional key additions.
    For the latter, we also discuss in \Cref{Sec: Bariant's attack} how key additions can be performed while maintaining the \Ciminion DRL Gr\"obner basis from \Cref{Th: Ciminion Groebner basis}.
    Additionally, in \Cref{Sec: Ciminion2} we propose \Ciminiontwo, a slightly modified variant which incorporates key additions that are compatible with \Cref{Th: Ciminion Groebner basis}.

    In \Cref{Sec: Ciminion cryptanalysis} we discuss Gr\"obner basis cryptanalysis of \Ciminion and \Ciminiontwo respectively.
    Due to our DRL Gr\"obner basis we yield a clean analysis as well as state-of-the-art complexity estimations for \Ciminion.
    In particular, we improve upon the designer's original analysis.
    In \Cref{Tab: Ciminion complexity summary} we present the minimal \Ciminion round numbers that achieve $128$ bits of security for an ideal adversary\footnote{In our setting an ideal adversary has access to a hypothetical matrix multiplication algorithm which achieves the linear algebra constant $\omega = 2$.} against the designer's estimation (``Fully Substituted Model''), Bariant's attack and our dedicated Eigenvalue Method for \Ciminion and \Ciminiontwo.

    \begin{table}[H]
        \centering
        \caption{Minimal number of rounds required for \Ciminion and \Ciminiontwo to achieve at least $128$ bits of security against various Gr\"obner basis attacks.
                 All computations are performed over the prime $q = 2^{127} + 45$ with $\omega = 2$.}
        \label{Tab: Ciminion complexity summary}
        \begin{tabular}{ c | M{25mm} | c || M{37mm} }
            \toprule
            & \multicolumn{3}{ c }{Complexity (bits)} \\
            \midrule
            $r_C + r_E$ & Bariant's Attack \cite{Bariant-Ciminion} & Eigenvalue Method & Fully Substituted Model \cite[\S 4.4]{EC:DGGK21} \\
            \midrule

            $33$  & $46.67$  & $63.09$  & $130$ \\
            $66$  & $81.22$  & $129.09$ & $262$ \\
            $112$ & $128.47$ & $221.09$ & $446$ \\

            \bottomrule
        \end{tabular}
    \end{table}

    \paragraph{\Hydra.}
    The heads of the \Hydra can be considered as iterated Lai--Massey ciphers, where the key as well as the input are hidden from the adversary and only the output is released.
    In a recent work Steiner \cite{ToSC:Steiner24} introduced the notion of polynomial systems in \emph{generic coordinates} \cite{Caminata-SolvingPolySystems} to cryptographic polynomial models.
    In essence, if a polynomial system is in generic coordinates, then we have a proven complexity for the computation of the DRL Gr\"obner basis \cite{Caminata-SolvingPolySystems}.
    In \cite[\S 6.2]{ToSC:Steiner24} the iterated polynomial models of generalized Feistel ciphers were analyzed for being in generic coordinates.
    For many instances of Feistel ciphers generic coordinates can be verified via the rank of an associated linear system.
    In \Cref{Th: Hydra generic coordinates test} we will construct such a linear system for the iterated \Hydra model.
    Hence, with the aid of a computer algebra program like \SageMath \cite{SageMath} we yield a computer-aided proof for the complexity of DRL Gr\"obner basis computations.
    Additionally, if the linear system for generic coordinates verification has full rank, then we can perform a linear change of coordinates to transform the \Hydra polynomial model into a DRL Gr\"obner basis, see \Cref{Sec: Extracting a Groebner Basis}.

    Via the change of coordinates we produce a quadratic DRL Gr\"obner basis in $2 \cdot r_\mathcal{H} - 2$ polynomials and variables together with four additional quadratic polynomials.
    This leaves us with two possible routes to recover the key:
    \begin{enumerate}[label=(\Alph*)]
        \item We can recompute the DRL Gr\"obner basis for $2 \cdot r_\mathcal{H} + 2$ quadratic polynomials.

        \item We ignore the additional equations and just use the DRL Gr\"obner basis to construct a univariate polynomial.
    \end{enumerate}

    The first approach was also considered by the \Hydra designers \cite[\S 7.4]{EC:GOSW23}, although without the linear transformation to the DRL Gr\"obner basis.
    Under the semi-regularity assumption \cite{Bardet-Complexity} the designers derived the \emph{degree of regularity}, which determines the complexity of the computation.
    In principle, we can apply the same strategy, but with one important benefit: After the linear transformation we only have to consider \emph{square-free} monomials in a DRL Gr\"obner basis computation.
    This effectively saves both time and space.
    In analogy to Gr\"obner basis computations over $\F_2$ we call this the \emph{Boolean Semi-Regular} Gr\"obner basis attack.

    Alternatively, we can simply restrict to the DRL Gr\"obner basis and proceed with term order conversion via the probabilistic FGLM algorithm \cite{Faugere-SubCubic}.
    It turns out that this approach invalidates the designers' claim that $r_\mathcal{H} = 31$ rounds cannot be attacked by an ideal adversary below $128$ bits.
    In addition, via a dedicated Eigenvalue Method an ideal adversary can attack up to $r_\mathcal{H} = 33$ rounds.

    In \Cref{Sec: Hydra cryptanalysis} we discuss the complexity estimations for \Hydra in more detail, a summary of the considered attacks is given in \Cref{Tab: Hydra summary}.
    Let $r_\mathcal{H}^\ast$ be the least round number which achieves $128$ bits of security against an ideal adversary.
    The \Hydra designers require for the head round number that $r_\mathcal{H} = \ceil{1.25 \cdot \max \left\{ 24, 2 + r_\mathcal{H}^\ast \right\}}$ \cite[\S 5.5]{EC:GOSW23}.
    In the original analysis the designers found that $r_\mathcal{H}^\ast = 29$ and $r_\mathcal{H} = 39$, however our analysis yields $r_\mathcal{H}^\ast = 34$ which implies the increase $r_\mathcal{H} = 45$ to maintain the originally intended security margin.

    \begin{table}[H]
        \centering
        \caption{Complexities of Gr\"obner basis attacks on \Hydra which beat the designers' analysis, and security margin reduction for full rounds.
                 All computations are performed over the prime $q = 2^{127} + 45$ with $\omega = 2$.}
        \label{Tab: Hydra summary}
        \begin{tabular}{ c | M{18mm} | M{16mm} | M{21mm} || M{21mm} }
            \toprule
            & \multicolumn{4}{ c }{Complexity (Bits)} \\
            \midrule
            $r_\mathcal{H}$ & Term Order Conversion & Eigenvalue Method & Boolean Semi-Regular Estimate & Semi-Regular Estimate \cite{EC:GOSW23} \\
            \midrule

            $29$ & $117.81$ & $111.09$ & $123.29$ & $130.80$ \\
            $31$ & $125.91$ & $119.09$ & $135.91$ & $141.77$ \\
            $33$ & $134.00$ & $127.09$ & $145.75$ & $152.75$ \\
            $39$ & $158.25$ & $151.09$ & $172.66$ & $182.22$ \\

            \bottomrule
        \end{tabular}
    \end{table}

    \paragraph{Implementation.}
    We implemented the \Ciminion and \Hydra polynomial systems and their DRL Gr\"obner bases in the computer algebra systems \SageMath \cite{SageMath} and \OSCAR \cite{OSCAR}.\footnote{\repository}
    In particular, for \Hydra we have implemented the linear system for the generic coordinates verification as well as the linear change of coordinates to produce a quadratic DRL Gr\"obner basis.
    We stress that our \SageMath implementations are able to handle practical parameters like $128$ bit prime numbers and full rounds \Hydra with $r_\mathcal{H} = 39$.

    \subsection{Organization of the Paper}
    In \Cref{Sec: Preliminaries} we introduce the technical requirements of this paper.
    \Ciminion and \Hydra are formally introduced in \Cref{Sec: Ciminion polynomial system,Sec: Hydra polynomial system} respectively.
    In \Cref{Sec: Groebner Bases} we quickly review Gr\"obner basis algorithms and the notion of generic coordinates.
    As preparation for \Hydra cryptanalysis, in \Cref{Sec: Boolean Macaulay matrices} we consider polynomial systems $\mathcal{F} \subset K [x_1, \dots, x_n]$ which contain a subset $\mathcal{F}_\text{Bool} = \{ g_1, \dots, g_n \} \subset \mathcal{F}$ with $\LM_{DRL} \left( g_i \right) = x_i^2$ for all $1 \leq i \leq n$.
    For such polynomial systems we can reduce Gaussian elimination on the Macaulay matrix $M_{\leq d}$ to a submatrix $M_{\leq d}^\text{Bool}$ which only represents the square-free monomials.
    The complexity of square-free Gaussian elimination is derived in \Cref{Th: Boolean Macaulay matrix}, and the one of the auxiliary division by remainder step is given in \Cref{Cor: Boolean Macaulay matrix construction complexity}.
    Finally, in \Cref{Sec: Eigenvalue Method} we develop a dedicated strategy, called the \emph{Eigenvalue Method}, to find a $\Fq$-valued key guess via the quadratic DRL Gr\"obner bases for \Ciminion and \Hydra.

    In \Cref{Sec: Analysis of Ciminion} we analyze the iterated \Ciminion polynomial system.
    In particular, in \Cref{Th: Ciminion Groebner basis} we construct a \Ciminion DRL Gr\"obner basis via linear transformations.
    Bariant's attack and compatibility of countermeasures with the DRL Gr\"obner basis are discussed in \Cref{Sec: Bariant's attack}.
    Finally, in \Cref{Sec: Ciminion cryptanalysis} we discuss Gr\"obner basis cryptanalysis of \Ciminion.

    In \Cref{Sec: Analysis of Hydra} we analyze the iterated polynomial system for the \Hydra heads.
    A linear system to verify generic coordinates for \Hydra is developed in \Cref{Th: Hydra generic coordinates test}.
    In case the linear system has full rank we discuss in \Cref{Sec: Extracting a Groebner Basis} that a linear change of coordinates yields a quadratic DRL Gr\"obner basis together with four additional quadratic polynomials.
    Note that we do not perform this change of coordinates by hand, instead we outsource it to our \SageMath or \OSCAR implementations.
    In \Cref{Sec: Hydra cryptanalysis} we finish this paper with Gr\"obner basis cryptanalysis for \Hydra.

    \section{Preliminaries}\label{Sec: Preliminaries}
    We denote fields with $K$ and their algebraic closure by $\overline{K}$.
    The finite field with $q$ elements is denoted as $\Fq$.
    If the field and the number of variables are clear from context, then we abbreviate the polynomial ring as $P = K [x_1, \dots, x_n]$.
    Matrices $\mathbf{M} \in K^{m \times n}$ are denoted with upper bold letters and vectors $\mathbf{v} \in K^n$ by lower bold letters.
    Matrix-vector and matrix-matrix products are written as $\mathbf{M} \mathbf{v}$.

    Let $f \in K [x_1, \dots, x_n]$ be a polynomial, and let $x_0$ be an additional variable, the homogenization of $f$ with respect to $x_0$ is defined as
    \begin{equation}
        f^\homog (x_0, \dots, x_n) = x_0^{\degree{f}} \cdot f \left( \frac{x_1}{x_0}, \dots, \frac{x_n}{x_0} \right) \in K [x_0, \dots, x_n].
    \end{equation}
    Analogously, we denote the homogenization of ideals $I^\homog = \left\{ f^\homog \mid f \in I \right\}$ and of finite systems of polynomials $\mathcal{F}^\homog = \left\{ f_1^\homog, \dots, f_m^\homog \right\}$.

    Let $f \in K [x_1, \dots, x_n]$ be a polynomial, we can always decompose it as sum $f = f_d + \ldots + f_0$, where $f_i$ is homogeneous of degree $i$.
    We call $f_d$ the homogeneous highest degree component of $f$ and denote it as $f^\topcomp$.
    Analogously, we denote $\mathcal{F}^\topcomp = \left\{ f_1^\topcomp, \dots, f_m^\topcomp \right\}$.
    Note that we have $f^\topcomp = f^\homog \mod \left( x_0 \right)$.

    Let $I \subset R$ be an ideal in a commutative ring, the radical of $I$ is defined as $\sqrt{I} = \left\{ f \in R \mid \exists n \in \Z_{\geq 1} \!: f^n \in I \right\}$.
    It is well-known that the radical is also an ideal.

    \subsection{\Ciminion}\label{Sec: Ciminion polynomial system}
    Let $\Fq$ be a finite field, the \Ciminion \cite{EC:DGGK21} PRF over $\Fq$ is most easily described via an illustration, see \Cref{Fig: Ciminion}.
    A nonce $\aleph \in \Fq$ and the first key pair $(K_1, K_2)$ are fed into the permutation $p_C: \Fq^3 \to \Fq^3$.
    Then the outputs are fed into the permutation $p_E: \Fq^3 \to \Fq^3$.
    The first two outputs of $p_E$ can then be added to messages $P_1, P_2 \in \Fq$ to encrypt them, and the third output remains private.
    If one wishes to encrypt an additional message pair $P_3, P_4 \in \Fq$ one feeds the outputs of $p_C$ after a key addition with $(K_3, K_4)$ into the rolling function $rol: \Fq^3 \to \Fq^3$ and then again into $p_E$.

    The permutations $p_C$ and $p_E$ are based on an iterated Feistel network
    \begin{equation}\label{Equ: Ciminion round function}
        \begin{split}
            \mathcal{R}^{(i)}: \Fq^3 &\to \Fq^3, \\
            \begin{pmatrix}
                x \\ y \\ z
            \end{pmatrix}
            &\mapsto
            \begin{pmatrix}
                0 & 0         & 1         \\
                1 & c_4^{(i)} & c_4^{(i)} \\
                0 & 1         & 1
            \end{pmatrix}
            \begin{pmatrix}
                x \\ y \\ z + x \cdot y
            \end{pmatrix}
            +
            \begin{pmatrix}
                c_1^{(i)} \\ c_2^{(i)} \\ c_3^{(i)}
            \end{pmatrix}
            ,
        \end{split}
    \end{equation}
    where $c_1^{(i)}, c_2^{(i)}, c_3^{(i)} \in \Fq$ and $c_4^{(i)} \in \Fq \setminus \{ 0, 1 \}$.
    Note that the Feistel network $(x, y, z)^\intercal \mapsto (x, y, z + x \cdot y)^\intercal$ is also known as Toffoli gate which plays a special role in reversible computing \cite{Toffoli-ReversibleComputing}.
    Finally, the rolling function is defined as
    \begin{equation}
        rol: \Fq^3 \to \Fq^3,
        \begin{pmatrix}
            x \\ y \\ z
        \end{pmatrix}
        \mapsto
        \begin{pmatrix}
            z + x \cdot y \\ x \\ y
        \end{pmatrix}
        .
    \end{equation}

    \clearpage 
    \begin{figure}[H]
        \centering
        \includegraphics[width=0.7\textwidth]{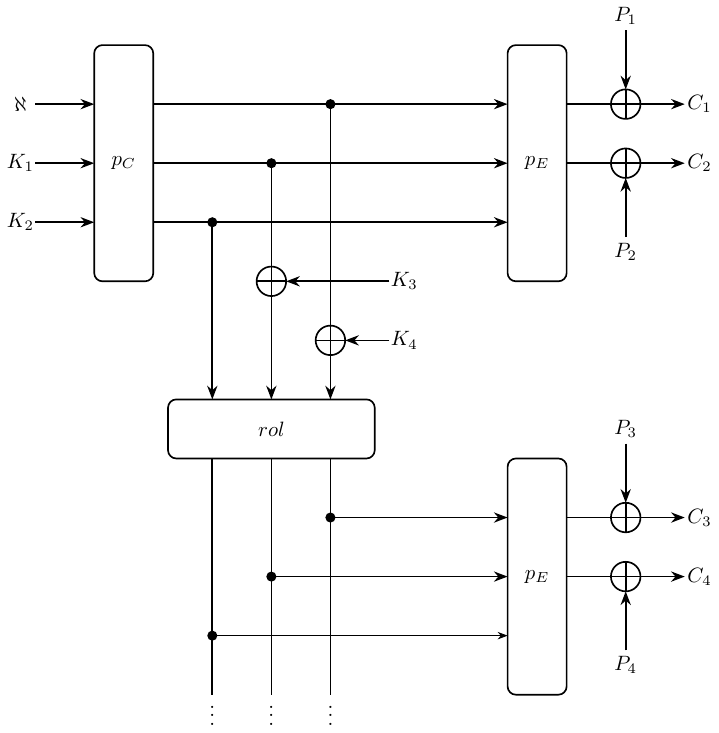}
        \caption{Encryption with \Ciminion.}
        \label{Fig: Ciminion}
    \end{figure}

    Formulas to compute the round numbers for \Ciminion are given in \cite[Table~1]{EC:DGGK21}. 
    E.g., for $s = 128$ and data limit $2^\frac{s}{2}$ we have the round numbers $r_C = 90$ and $r_E = 14$.

    It is now straight-forward to set up the iterated polynomial system for the first key pair $(K_1, K_2)$ given the first plain/ciphertext pair.
    For ease of notation we abbreviate $r = r_C + r_E$.
    \begin{defn}[Iterated polynomial system for \Ciminion]\label{Def: Ciminion polynomial model}
        Let $\Fq$ be a finite field, and let $\big( \aleph, (p_1, p_2), (c_1, c_2) \big) \in \Fq \times \Fq^2 \times\Fq^2$ be a nonce/plain/ciphertext sample given by a \Ciminion encryption function.
        We call the polynomial system $\mathcal{F}_\Ciminion = \left\{ \mathbf{f}^{(i)} \right\}_{1 \leq i \leq r} \subset \Fq \Big[ y_1, y_2, \mathbf{x}^{(1)}, \dots, \mathbf{x}^{(r - 1)}, x \Big]$ the iterated \Ciminion polynomial system

        \[
    	    \mathbf{f}^{(i)}
    	    =
    	    \begin{dcases}
    	        \mathcal{R}^{(1)} \left( \aleph, y_1, y_2 \right) - \mathbf{x}^{(1)}, & i = 1, \\
    	        \mathcal{R}^{(i)} \left( \mathbf{x}^{(i - 1)} \right) - \mathbf{x}^{(i)}, & 2 \leq i \leq r - 1, \\
    	        \mathcal{R}^{(r)} \left( \mathbf{x}^{(r - 1)} \right) -
    	        \begin{pmatrix}
    	            c_1 - p_1 \\
    	            c_2 - p_2 \\
    	            x
    	        \end{pmatrix}
    	        , & i = r.
    	    \end{dcases}
        \]
    \end{defn}

    \subsection{\Hydra}\label{Sec: Hydra polynomial system}
    Let $\Fq$ be a finite field, the \Hydra \cite{EC:GOSW23} PRF over $\Fq$ is also easily described via an illustration, see \Cref{Fig: Hydra}.
    A nonce $\aleph \in \Fq$ and an initial value $\texttt{IV} \in \Fq^3$ are added to a key $\texttt{K} \in \Fq^4$, and then fed into the keyed \emph{body} function $\mathcal{B}: \Fq^4 \times \Fq^4 \to \Fq^8$ which produces the state $\mathbf{y} || \mathbf{z} = (\mathbf{y}, \mathbf{z})^\intercal \in \Fq^4 \times \Fq^4 $.
    Then $\mathbf{y} || \mathbf{z}$ is fed into the first keyed \emph{head} function $\mathcal{H}_\texttt{K}: \Fq^8 \times \Fq^4 \to \Fq^8$.
    After adding $\mathbf{y} || \mathbf{z}$ to the output of the head the result is released as the first \Hydra sample.
    Say one has already requested $(i - 1)$ \Hydra samples, for an additional sample $\mathbf{y} || \mathbf{z}$ is fed into the $i$\textsuperscript{th} rolling function $\mathcal{R}_i: \Fq^8 \to \Fq^8$, and then again into the head $\mathcal{H}_\texttt{K}$.
    After adding the output of $\mathcal{R}_i$ to the output of $\mathcal{H}_\texttt{K}$ the result is released as $i$\textsuperscript{th} \Hydra sample.
    \begin{figure}[H]
        \centering
        \includegraphics[width=0.7\textwidth]{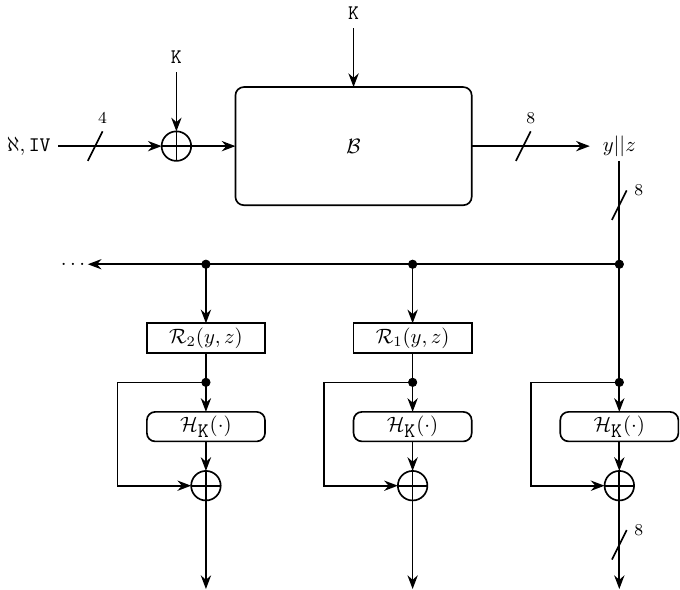}
        \caption{The \Hydra PRF with hidden body.}
        \label{Fig: Hydra}
    \end{figure}

    For our polynomial model we are only going to consider the heads of the \Hydra.
    Thus, we do not describe the details of the body function $\mathcal{B}$, interested readers can find the description in \cite[\S 5.2]{EC:GOSW23}.
    Let $\mathbf{M}_\mathcal{E} \in \GL{4}{\Fq}$ be an MDS matrix, e.g.\ the \Hydra designers recommend $\circulant \left( 2, 3, 1, 1 \right)$ or $\circulant \left( 3, 2, 1, 1 \right)$.\footnote{
        We consider circulant matrices with respect to a right shift, i.e.\
        \[
            \circulant (a_1, \dots, a_n) =
            \begin{pmatrix}
                a_1    & a_2    & \ldots  & a_{n - 1} & a_n       \\
                a_n    & a_1    & \ldots  & a_{n - 2} & a_{n - 1} \\
                \vdots & \vdots & \ddots  & \vdots    & \vdots    \\
                a_2    & a_3    & \ldots  & a_n       & a_1
            \end{pmatrix}
            .
        \]}
    For the keyed permutation $\mathcal{H}_\texttt{K}$, let
    \begin{align}
        \texttt{K}' &=
        \begin{pmatrix}
            \texttt{K} \\ \mathbf{M}_\mathcal{E} \texttt{K}
        \end{pmatrix}
        \in \Fq^8, \\
        f (x_1, \dots, x_n) &= \left( \sum_{i = 1}^{8} (-1)^{\floor{\frac{i - 1}{4}}} \cdot x_i \right)^2.
    \end{align}
    Moreover, let $\mathbf{M}_\mathcal{J} \in \GL{8}{\Fq}$ be an invertible matrix.
    The heads round function is defined as
    \begin{equation}
        \begin{split}
            \mathcal{J}_i : \Fq^8 \times \Fq^4 &\to \Fq^8, \\
            \left( \mathbf{x}, \texttt{K} \right) &\mapsto \mathbf{M}_\mathcal{J}
            \big( x_i + f (\mathbf{x}) \big)_{1 \leq i \leq 8} + \texttt{K}' + \mathbf{c}^{(i)},
        \end{split}
    \end{equation}
    where $\mathbf{c}^{(i)} \in \Fq^8$ is a constant.
    The heads function is now defined as the following $r_\mathcal{H}$-fold composition
    \begin{equation}
        \begin{split}
            \mathcal{H}_\texttt{K}: \Fq^4 \times \Fq^4 \times \Fq^4 & \to \Fq^8, \\
            (\mathbf{y}, \mathbf{z}, \texttt{K}) &\mapsto \mathcal{J}_{r_\mathcal{H}} \circ \cdots \circ \mathcal{J}_1 \left(
            \begin{pmatrix}
                \mathbf{y} \\ \mathbf{z}
            \end{pmatrix}
            , \texttt{K} \right).
        \end{split}
    \end{equation}
    Let $\mathbf{M}_\mathcal{I} \in \GL{4}{\Fq}$ be an invertible matrix of the form
    \begin{equation}
        \mathbf{M}_\mathcal{I} =
        \begin{pmatrix}
            \mu_{1, 1} & 1 & 1 & 1 \\
            \mu_{2, 1} & \mu_{2, 2} & 1 & 1 \\
            \mu_{3, 1} & 1 & \mu_{3, 3} & 1 \\
            \mu_{4, 1} & 1 & 1 & \mu_{4, 4}
        \end{pmatrix}
        .
    \end{equation}
    Note that the $\mu_{i, j} \in \Fq$ have to be chosen under additional constraints to prevent invariant subspaces \cite[\S 5.2]{EC:GOSW23}.
    Finally, for the rolling function let
    \begin{align}
        g (\mathbf{y}, \mathbf{z}) &= \left( \sum_{i = 1}^{4} (-1)^{i - 1} \cdot y_i \right) \cdot \left( \sum_{i = 1}^{4} (-1)^{\floor{\frac{i - 1}{2}}} \cdot z_i \right), \\
        \mathbf{M}_\mathcal{R} &= \diag \left( \mathbf{M}_\mathcal{I}, \mathbf{M}_\mathcal{I} \right).
    \end{align}
    Then, the rolling function is defined as
    \begin{equation}
        \begin{split}
            \mathcal{R}: \Fq^4 \times \Fq^4 &\to \Fq^8, \\
            \begin{pmatrix}
                \mathbf{y} \\ \mathbf{z}
            \end{pmatrix}
            &\mapsto \mathbf{M}_\mathcal{R}
            \begin{pmatrix}
                \big( y_i + g (\mathbf{y}, \mathbf{z}) \big)_{1 \leq i \leq 4} \\
                \big( z_i + g (\mathbf{z}, \mathbf{y}) \big)_{1 \leq i \leq 4}
            \end{pmatrix}
            ,
        \end{split}
    \end{equation}
    and the $i$\textsuperscript{th} rolling function is simply the composition $\mathcal{R}_i = \mathcal{R} \circ \mathcal{R}_{i - 1} + \mathbf{c}_\mathcal{R}^{(i)}$, where $\mathbf{c}_\mathcal{R}^{(i)} \in \Fq$ is a random constant.
    Note that the degree increasing parts in $\mathcal{J}_i$ and $\mathcal{R}$ are generalized Lai--Massey permutations \cite{ToSC:GOPS22,Steiner-GTDS}.

    For $\kappa$ bits of security, the number of rounds for the \Hydra heads with data limit $2^\frac{\kappa}{2}$ are computed via, see \cite[\S 5.5]{EC:GOSW23},
    \begin{equation}\label{Equ: head round numbers}
        r_\mathcal{H} = \ceil{1.25 \cdot \max \left\{ 24, 2 + r_\mathcal{H}^\ast \right\}},
    \end{equation}
    where $r_\mathcal{H}^\ast$ is the smallest integer that satisfies \cite[Equation~12]{EC:GOSW23} and \cite[Equation~15]{EPRINT:GOSW22}.
    Note that \cite[Equation~12]{EC:GOSW23} is derived from the Gr\"obner basis cryptanalysis of \Hydra.
    E.g., for $q = 2^{127} + 45$ and $\kappa = 128$ the $\Hydra$ designers propose $r_\mathcal{H} = 39$.

    Given two \Hydra outputs it is straight-forward to set up an overdetermined polynomial system for \Hydra.
    For ease of writing we denote the $i$\textsuperscript{th} intermediate state variables for the $j$\textsuperscript{th} output as $\mathbf{x}^{(i)}_j = \left( x_{j, 1}^{(i)}, \dots, x_{j, 8}^{(i)} \right)^\intercal$ and analogously for $\mathbf{y} = (y_1, \dots, y_4)^\intercal$, $\mathbf{z} = (z_1, \dots, z_4)^\intercal$ and $\mathbf{k} = (k_1, \dots, k_4)^\intercal$.
    \begin{defn}[Iterated polynomial system for \Hydra]\label[defn]{Def: Hydra polynomial system}
        Let $\Fq$ be a finite field, and let $\mathbf{c}_1, \mathbf{c}_2 \in \Fq^8$ be the first two outputs of a \Hydra function.
        We call the following polynomial system $\mathcal{F}_\Hydra = \left\{ \mathbf{f}_1^{(i)}, \mathbf{f}_\mathcal{R}, \mathbf{f}_2^{(i)} \right\}_{1 \leq i \leq r_\mathcal{H}} \subset \Fq \Big[ \mathbf{y}, \mathbf{z}, \mathbf{x}_1^{(1)}, \dots, \mathbf{x}_1^{(r_\mathcal{H} - 1)}, \mathbf{x}_2^{(0)}, \dots, \mathbf{x}_2^{(r_\mathcal{H} - 1)}, \mathbf{k} \Big]$ the iterated \Hydra polynomial system
        \begin{align}
            \mathbf{f}_1^{(i)} &=
            \begin{dcases}
                \mathcal{J}_1 \left(
                \begin{pmatrix}
                    \mathbf{y} \\ \mathbf{z}
                \end{pmatrix}
                , \mathbf{k}'
                \right) - \mathbf{x}_1^{(1)}, & i = 1, \\
                \mathcal{J}_i \left( \mathbf{x}_1^{(i - 1)}, \mathbf{k}' \right) - \mathbf{x}_1^{(i)}, & 2 \leq i \leq r_\mathcal{H} - 1, \\
                \mathcal{J}_{r_\mathcal{H}} \left( \mathbf{x}_1^{(r_\mathcal{H} - 1)}, \mathbf{k}' \right) +
                \begin{pmatrix}
                    \mathbf{y} \\ \mathbf{z}
                \end{pmatrix}
                - \mathbf{c}_1, & i = r_\mathcal{H},
            \end{dcases}
            \nonumber \\
            \mathbf{f}_\mathcal{R} &= \mathcal{R} \left(
            \begin{pmatrix}
                \mathbf{y} \\ \mathbf{z}
            \end{pmatrix}
            \right) + \mathbf{c}_\mathcal{R}^{(1)} - \mathbf{x}_2^{(0)},  \nonumber \\
            \mathbf{f}_2^{(i)} &=
            \begin{dcases}
                \mathcal{J}_1 \left( \mathbf{x}_2^{(0)}, \mathbf{k}' \right) - \mathbf{x}_2^{(1)}, & i = 1, \\
                \mathcal{J}_i \left( \mathbf{x}_2^{(i - 1)}, \mathbf{k}' \right) - \mathbf{x}_2^{(i)}, & 2 \leq i \leq r_\mathcal{H} - 1, \\
                \mathcal{J}_{r_\mathcal{H}} \left( \mathbf{x}_2^{(r_\mathcal{H} - 1)}, \mathbf{k}' \right) + \mathbf{x}_2^{(0)} - \mathbf{c}_2, & i = r_\mathcal{H},
            \end{dcases}
            \nonumber
        \end{align}
        where $\mathbf{k}' =
        \begin{pmatrix}
            \mathbf{k} \\ \mathbf{M}_\mathcal{E} \mathbf{k}
        \end{pmatrix}
        .$
    \end{defn}
    This polynomial system consists of $16 \cdot r_\mathcal{H} + 8$ many polynomials in $16 \cdot r_\mathcal{H} + 4$ many variables.
    It can also easily be extended to an arbitrary number of output samples.

    \subsection{Gr\"obner Bases}\label{Sec: Groebner Bases}
    Gr\"obner bases are a fundamental tool in computational algebra.
    They were first introduced by Bruno Buchberger in his PhD thesis \cite{Buchberger}.
    For their definition one needs to induce an order on the monomials of a polynomial ring $P = K [x_1, \dots, x_n]$.
    First notice that a monomial $m = \prod_{i = 1}^{n} x_i^{a_i} \in P$ can be identified with its exponent vector $\mathbf{a} = \left( a_1, \dots, a_n \right)^\intercal \in \mathbb{Z}_{\geq 0}^{n}$.
    Therefore, a term order $>$ on $P$ is a binary relation on $\mathbb{Z}_{\geq 0}^n$ such that, see \cite[Chapter~2~\S 2~Definition~1]{Cox-Ideals}:
    \begin{enumerate}[label=(\roman*)]
        \item $>$ is a total ordering in $\mathbb{Z}_{\geq 0}^n$.

        \item If $\mathbf{a} > \mathbf{b}$ and $\mathbf{c} \in \mathbb{Z}_{\geq 0}^n$, then $\mathbf{a} + \mathbf{c} > \mathbf{b} + \mathbf{c}$.

        \item Every non-empty subset of $\mathbb{Z}_{\geq 0}^n$ has a smallest element under $>$, i.e.\ $>$ is a well-ordering on $\mathbb{Z}_{\geq 0}^n$.
    \end{enumerate}

    Now let $I \subset P = K [x_1, \dots, x_n]$ be an ideal, and let $>$ be a term order on $P$.
    For $f \in P$ we denote with $\LM_> \left( f \right)$ the leading monomial of $f$, i.e.\ the largest monomial with non-zero coefficient in $f$ under $>$.
    A finite ideal basis $\mathcal{G} = \{ g_1, \dots, g_m \} \subset I$ such that
    \begin{equation}
        \big( \LM_> (g_1), \dots, \LM_> (g_m) \big) = \big( \LM_> (f) \; \big\vert \; f \in I \big)
    \end{equation}
    is called a \emph{$>$-Gr\"obner basis} of $I$.
    I.e., the leading monomials of $\mathcal{G}$ are minimal under $>$.

    We abbreviate the ideal of leading terms of $I$ as $\LM_> (I) = \big( \LM_> (f) \; \big\vert \; f \in I \big)$.
    Also, for ease of writing we often denote the variables of $P$ in vector form $\mathbf{x} = (x_1, \dots, x_n)^\intercal$.
    If not specified otherwise, then we always assume that the variables within a variable vector $\mathbf{x}$ are naturally ordered with respect to $>$, i.e.\ $x_1 > \ldots > x_n$.

    With Gr\"obner bases one can solve many important computational problems like the membership problem or the computation of the set of zeros of an ideal.
    For a general introduction to Gr\"obner bases we refer to \cite{Kreuzer-CompAlg1,Kreuzer-CompAlg2,Cox-Ideals}.

    In this paper we work with the lexicographic (LEX) and the degree reverse lexicographic (DRL) term order.
    For $\mathbf{a}, \mathbf{b} \in \mathbb{Z}_{\geq 0}^n$, we have that $\mathbf{a} >_{LEX} \mathbf{b}$ if the first non-zero entry of $\mathbf{a} - \mathbf{b}$ is positive.
    We have that $\mathbf{a} >_{DRL} \mathbf{b}$ if either $\sum_{i = 1}^{n} a_i > \sum_{i = 1}^{n} b_i$ or $\sum_{i = 1}^{n} a_i = \sum_{i = 1}^{n} b_i$ and the last non-zero entry of $\mathbf{a} - \mathbf{b}$ is negative.

    In more generality, a term order $>$ on $P$ is called \emph{degree compatible} if for any $f, g \in P$ the degree inequality $\degree{f} > \degree{g}$ implies that $f > g$.
    Obviously, the DRL term order is degree compatible.

    \subsubsection{Linear Algebra-Based Gr\"obner Basis Algorithms}
    Gr\"obner bases can be constructed via Gaussian elimination on the so-called \emph{Macaulay matrices}.
    Let $\mathcal{F} \subset P = K [ x_1, \dots, x_n ]$ be homogeneous, and let $>$ be a term order on $P$.
    The \emph{homogeneous Macaulay matrix} $M_d$ for $\mathcal{F}$ has columns indexed by the monomials of degree $d$ in $P$ sorted via $>$, and its rows are indexed by $s \cdot f$, where $f \in \mathcal{F}$ and $s \in P$ is a monomial such that $\degree{s \cdot f} = d$.
    The entry of row $s \cdot f$ at column $t$ is the coefficient of $t$ in the polynomial $s \cdot f$.
    If $\mathcal{F}$ is inhomogeneous, then one replaces the degree equalities by inequalities to obtain the \emph{inhomogeneous Macaulay matrix} $M_{\leq d}$.
    Obviously, if we perform Gaussian elimination on $M_0, \dots, M_d$ (respectively $M_{\leq d}$) for $d$ large enough, then we construct a $>$-Gr\"obner basis for $\mathcal{F}$.
    This idea can be traced back to Lazard \cite{Lazard-Groebner}.
    The least $d$ which yields a Gr\"obner basis is of special interest.
    \begin{defn}[{Solving degree, \cite[Definition~6]{Caminata-SolvingPolySystems}}]\label[defn]{Def: solving degree}
        Let $\mathcal{F} = \{ f_1, \dots, f_m \} \subset K [x_1, \dots, \allowbreak x_n]$ and let $>$ be a term order.
        The solving degree of $\mathcal{F}$ is the least degree $d$ such that Gaussian elimination on the Macaulay matrix $M_{\leq d}$ produces a Gr\"obner basis of $\mathcal{F}$ with respect to $>$. We denote it by $\solvdeg_> (\mathcal{F})$.

        If $\mathcal{F}$ is homogeneous, then we consider the homogeneous Macaulay matrix $M_d$ and let the solving degree of $\mathcal{F}$ be the least degree $d$ such that Gaussian elimination on $M_0, \dots, M_d$ produces a Gr\"obner basis of $\mathcal{F}$ with respect to $>$.
    \end{defn}

    Algorithms which perform Gaussian elimination on the Macaulay matrices are generally referred to as \emph{linear algebra-based Gr\"obner basis algorithms}.
    Modern examples of such algorithms are Faug\`{e}re's F4 \cite{Faugere-F4} and Matrix-F5 \cite{Faugere-F5} algorithms.
    Since the solving degree is a priori unknown for most polynomial systems, F4/5 construct the Macaulay matrices for increasing values of $d$.
    Therefore, they require a stopping criterion to decide whether a Gr\"obner basis has already been found.
    On the other hand, we can always stop F4/5 after the final Gaussian elimination in the solving degree, since the row space basis must already contain a Gr\"obner basis.
    For this reason, we also consider the solving degree as complexity measure for advanced linear algebra-based algorithms.

    \subsubsection{Polynomial Systems in Generic Coordinates}
    Polynomial systems in generic coordinates have a proven DRL solving degree upper bound.
    For their definition we recall the \emph{degree of regularity} of a polynomial system.
    \begin{defn}[{Degree of regularity, \cite[Definition~4]{Bardet-Complexity}}]\label[defn]{Def: degree of regularity}
        Let $K$ be a field, and let $\mathcal{F} \subset K [x_1, \dots, x_n]$.
        If $\dim \left( \mathcal{F}^\topcomp \right) = 0$, then the degree of regularity is defined as
        \[
            d_{\reg} \left( \mathcal{F} \right) = \min \left\{ d \geq 0 \; \middle\vert \; \dim_{\Fq} \Big( \big\{ f \in \left( \mathcal{F}^\topcomp \right) \; \big\vert \; \degree{f} = d \big\} \Big) = \binom{n + d - 1}{d} \right\}.
        \]
        If $\dim \left( \mathcal{F}^\topcomp \right) > 0$, then $d_{\reg} \left( \mathcal{F} \right) = \infty$.
    \end{defn}

    We present an equivalent definition for generic coordinates which requires slightly less knowledge about commutative algebra and projective geometry than the original definition from \cite[Definition~5]{Caminata-SolvingPolySystems,Caminata-SolvingPolySystemsPreprint}.
    \begin{defn}[{cf.\ \cite[Definition~5]{Caminata-SolvingPolySystems,Caminata-SolvingPolySystemsPreprint}, \cite[Theorem~3.2]{ToSC:Steiner24}}]\label[defn]{Def: generic coordinates}
        Let $K$ be an infinite field, and let $\mathcal{F} \subset P = K [x_1, \dots, x_n]$.
        We say that $\left( \mathcal{F}^\homog \right) \subset P [x_0]$ is in generic coordinates if either $d_{\reg} \left( \mathcal{F} \right) < \infty$ or $\sqrt{\mathcal{F}^\homog} = (x_0, \dots, x_n)$.

        Let $K$ be any field, and let $K \subset L$ be an infinite field extension.
        $I$ is in generic coordinates over $K$ if $I \otimes_K L [x_0, \dots, x_n] \subset L [x_0, \dots, x_n]$ is in generic coordinates.
    \end{defn}

    Let $\mathcal{F} \subset K [x_1, \dots, x_n]$ be an inhomogeneous polynomial system, to prove that $\left( \mathcal{F}^\homog \right)$ is in generic coordinates one can use the following procedure based on \cite[Theorem~3.2]{ToSC:Steiner24}:
    \begin{enumerate}
        \item Compute $\left( \mathcal{F}^\topcomp \right) = \left( \mathcal{F}^\homog \right) \mod (x_0)$.

        \item\label{Item: Step 2} Construct a monomial $x_i^d \in \left( \mathcal{F}^\topcomp \right)$ for some $1 \leq i \leq n$ and $d \geq 1$.

        \item Replace $\left( \mathcal{F}^\topcomp \right) = \left( \mathcal{F}^\topcomp \right) \mod (x_i)$ and return to Step \ref{Item: Step 2}.
    \end{enumerate}
    If Step \ref{Item: Step 2} is successful for all $1 \leq i \leq n$, then the polynomial system $\mathcal{F}^\homog$ is in generic coordinates.

    As already mentioned, polynomial systems in generic coordinates have a proven DRL solving degree upper bound.
    \begin{thm}[{\cite[Theorem~9,~10,~Corollary~2]{Caminata-SolvingPolySystems}}]\label{Th: solvdeg and CM-regularity}
        Let $K$ be an algebraically closed field, and let $\mathcal{F} = \{ f_1, \dots, f_m \} \subset K [x_1, \dots, x_n]$ be an inhomogeneous polynomial system such that $\left( \mathcal{F}^\homog \right)$ is in generic coordinates.
        Let $d_i = \deg \left( f_i \right)$ and assume that $d_1 \geq \ldots \geq d_m$, then for $l \in \min \{ n + 1, m \}$
        \[
            \solvdeg_{DRL} \left( \mathcal{F} \right) \leq d_1 + \ldots + d_l - l + 1.
        \]
    \end{thm}

    In the Gr\"obner basis and algebraic cryptanalysis literature, the latter inequality is often referred to as the \emph{Macaulay bound}.

    \subsubsection{Complexity Estimates via the Solving Degree}
    It is well-known that the complexity of Gaussian elimination of a matrix $\mathbf{A} \in K^{n \times m}$ with $r = \rank \left( \mathbf{A} \right)$ is given by $\mathcal{O} \left( n \cdot m \cdot r^{\omega - 2} \right)$ field operations \cite[\S 2.2]{Storjohann-Matrix}, where $2 \leq \omega < 2.371552$ is a linear algebra constant \cite{SODA:WXXZ24}.

    Given an inhomogeneous polynomial system $\mathcal{F} \subset P = K [x_1, \dots, x_n]$, it is well-known that the number of monomials in $P$ of degree $d$ is $\binom{n + d - 1}{d}$.
    Thus, the number of columns of $M_{\leq d}$ is bounded by $\sum_{i = 0}^{d} \binom{n + i - 1}{i} = \binom{n + d}{d}$, and the number of rows by $\abs{\mathcal{F}} \cdot \binom{n + d}{d}$.
    So, the complexity of Gaussian elimination on the inhomogeneous Macaulay matrices $M_{\leq d}$ is bounded by
    \begin{equation}\label{Equ: Groebner basis complexity}
        \mathcal{O} \left( \abs{\mathcal{F}} \cdot \binom{n + d}{d}^\omega \right).
    \end{equation}

    \subsection{Boolean Macaulay Matrices}\label{Sec: Boolean Macaulay matrices}
    Let $\mathcal{F} \subset \F_2 [x_1, \dots, x_n]$ be a polynomial system.
    If we are only interested in $\F_2$-valued solutions, then it is sufficient to compute a Gr\"obner basis of the ideal $\big( \mathcal{F}, x_1^2 - x_1, \dots, x_n^2 - x_n \big)$.
    On the other hand, via the field equations $F = \left( x_1^2 - x_1, \dots, x_n^2 - x_n \right)$ we can reduce any monomial in the Macaulay matrix to a square-free monomial.
    Effectively, this reduces the full Macaulay matrix to a square-free submatrix, the so-called \emph{Boolean Macaulay matrix}.
    Gaussian elimination for Boolean Macaulay matrices with respect to DRL was analyzed in \cite{Bardet-Boolean} by Bardet et al.

    Interestingly, for \Ciminion and \Hydra we will construct zero-dimensional DRL Gr\"obner bases in \Cref{Sec: Ciminion Groebner basis,Sec: Extracting a Groebner Basis} with leading monomials  $x_i^2$.
    Hence, for Gr\"obner basis computations we can use these bases to reduce the Macaulay matrix to a ``Boolean'' one over an arbitrary finite field $\Fq$.
    This motivates the following definitions.
    \begin{defn}
        Let $K$ be a field, let $\mathcal{F} \subset P = K [x_1, \dots, x_n]$, and let $>$ be a degree compatible term order on $P$.
        $\mathcal{F}$ is said to be a $K$-Boolean polynomial system if
        \begin{enumerate}[label=(\roman*)]
            \item for every $1 \leq i \leq n$ there exists $f \in \mathcal{F}$ such that $\LM_> \left( f \right) = x_i^2$, and

            \item\label{Item: Groebner basis} $\mathcal{F}$ is a $>$-Gr\"obner basis.
        \end{enumerate}
    \end{defn}

    Note that Condition \ref{Item: Groebner basis} is introduced so that the remainder of multivariate polynomial division does not depend on the ordering of the elements of $\mathcal{F}$, see \cite[Chapter~2~\S 6~Proposition~1]{Cox-Ideals}.
    \begin{defn}[{cf.\ \cite[Definition~1]{Bardet-Boolean}}]
        Let $K$ be a field, let $\mathcal{F}_\text{Bool}, \mathcal{F} \subset P = K [x_1, \dots, x_n]$ be polynomial systems, and let $>$ be a degree compatible term order on $P$.
        Assume that $\mathcal{F}_\text{Bool}$ is $K$-Boolean with respect to $>$.
        For the (inhomogeneous) $K$-Boolean Macaulay matrix $M_{\leq d}^\text{Bool}$ of $\mathcal{F} \cup \mathcal{F}_\text{Bool}$:
        \begin{itemize}
            \item Columns are indexed by square-free monomials $s \in P \setminus \LM_> \left( \mathcal{F}_\text{Bool} \right)$ with $\degree{m} \leq d$ sorted via $>$.

            \item Rows are indexed by polynomials $t \cdot f$ where $t \in P \setminus \LM_> \left( \mathcal{F}_\text{Bool} \right)$ is a square-free monomial and $f \in \mathcal{F}$ such that $\degree{t \cdot f} \leq d$.

            \item In the row $t \cdot f$ one writes the coefficients of the remainder $r_{t \cdot f} = t \cdot f \mod \left( \mathcal{F}_\text{Bool} \right)$ with respect to $>$.
        \end{itemize}
    \end{defn}

    Note that any monomial $t \in P$ can be reduced to a square-free polynomial $\hat{t} \equiv t \mod \left( \mathcal{F}_\text{Bool} \right)$ with $\degree{\hat{t}} \leq \degree{t}$ due to degree compatibility of $>$.
    Since remainders with respect to a Gr\"obner basis are unique, see \cite[Chapter~2~\S 6~Proposition~1]{Cox-Ideals}, for the remainder of $t \cdot f$ it does not matter whether we reduce $t$ first and then $\hat{t} \cdot f$ or $t \cdot f$ at once.
    So it is indeed sufficient to consider only products with square-free monomials.

    Obviously, the remainder $r_{t \cdot f} = t \cdot f \mod \left( \mathcal{F}_\text{Bool} \right)$ is a polynomial consisting of only square-free terms.
    Moreover, by elementary properties of multivariate polynomial division with respect to a degree compatible term order, see \cite[Chapter~2~\S 3~Theorem~3]{Cox-Ideals}, we have that
    \begin{equation}
        t \cdot f = \sum_{i = 1}^{m} g_i \cdot f_i^\text{Bool} + r_{t \cdot f},
    \end{equation}
    where $\degree{g_i \cdot f_i^\text{Bool}}, \degree{r_{t \cdot f}} \leq \degree{t \cdot f}$.
    I.e., the $K$-Boolean Macaulay matrix can be constructed via row operations from the Macaulay matrix $M_{\leq d}$ of $\mathcal{F} \cup \mathcal{F}_\text{Bool}$.
    Alternatively, we can obtain the $K$-Boolean Macaulay matrix by pivoting non-square-free monomials in $M_{\leq d}$ via rows coming from $\mathcal{F}_\text{Bool}$, and then performing an auxiliary Gaussian elimination step that stops after eliminating all non-square-free monomials in the rows coming from $\mathcal{F}$.

    Let $M_{\leq d}$ be a Macaulay matrix, we denote with $\rowspace \left( M_{\leq d} \right)$ the row space basis, i.e.\ the polynomial system obtained via Gaussian elimination on $M_{\leq d}$.
    Analogously, we use this notation for Boolean Macaulay matrices.
    \begin{defn}
        Let $K$ be a field, let $\mathcal{F}_\text{Bool}, \mathcal{F} \subset P = K [x_1, \dots, \allowbreak x_n]$ be polynomial systems, and let $>$ be a degree compatible term order on $P$.
        Assume that $\mathcal{F}_\text{Bool}$ is $K$-Boolean with respect to $>$.
        The $K$-Boolean solving degree of $\mathcal{F} \cup \mathcal{F}_\text{Bool}$ is the least degree $d$ such that $\rowspace \left( M_{\leq d}^\text{Bool} \right) \cup \mathcal{F}_\text{Bool}$ contains a Gr\"obner basis with respect to $>$.
        We denote it by $\solvdeg^\text{Bool}_> \left( \mathcal{F} \cup \mathcal{F}_\text{Bool} \right)$.
    \end{defn}

    Note that $\rowspace \left( M_{\leq d}^\text{Bool} \right)$ might not contain a zero-dimensional Gr\"obner basis of the full system.
    A trivial example is $\mathcal{F} = \left\{ x_1^2, x_1 \cdot x_2, x_2^2 \right\}$.
    Therefore, to extract the Gr\"obner basis from $\rowspace \left( M_{\leq d}^\text{Bool} \right) \cup \mathcal{F}_\text{Bool}$ we set up the ideal of $>$-leading terms and then compute its minimal generating set.
    With the minimal generating set we can then filter the $>$-Gr\"obner basis.\footnote{
        It is obvious that all leading terms of $\rowspace \left( M_{\leq d} \right)$ not coming from $\rowspace \left( M_{\leq d}^\text{Bool} \right)$ are already contained in $\LM_{>} \left( \mathcal{F}_\text{Bool} \right)$.
        Since $\mathcal{F}_\text{Bool}$ is a $>$-Gr\"obner basis it is indeed sufficient to just consider $\rowspace \left( M_{\leq d}^\text{Bool} \right) \cup \mathcal{F}_\text{Bool}$.
    }

    The next theorem follows straight-forward from $M_{\leq d}^\text{Bool} \subset M_{\leq d}$ (as $K$-vector spaces).
    \begin{thm}\label{Th: Boolean Macaulay matrix}
        Let $K$ be a field, let $\mathcal{F}_\text{Bool}, \mathcal{F} \subset P = K [x_1, \dots, x_n]$ be polynomial systems, and let $>$ be a degree compatible term order on $P$.
        Let $D  = \min_{f \in \mathcal{F}} \degree{f}$, and let $d = \solvdeg^\text{Bool}_> \left( \mathcal{F} \cup \mathcal{F}_\text{Bool} \right)$.
        Assume that $\mathcal{F}_\text{Bool}$ is $K$-Boolean with respect to $>$.
        Then
        \begin{enumerate}
            \item $\solvdeg^\text{Bool}_> \left( \mathcal{F} \cup \mathcal{F}_\text{Bool} \right) \leq \solvdeg_> \left( \mathcal{F} \cup \mathcal{F}_\text{Bool} \right)$.

            \item If $d - D < n$, then Gaussian elimination on $M_{\leq d}^\text{Bool}$ requires
            \[
            	\mathcal{O} \left( \abs{\mathcal{F}} \cdot \left( \sum_{i = 0}^{d - D} \binom{n}{i} \right) \cdot \left( \sum_{i = 0}^{d} \binom{n}{i} \right)^{\omega - 1} \right)
            \]
            field operations.

            \item If $d - D \geq n$, then Gaussian elimination on $M_{\leq d}^\text{Bool}$ requires at most
            \[
            	\mathcal{O} \left( \abs{\mathcal{F}} \cdot 2^{\omega \cdot n} \right)
            \]
            field operations.
        \end{enumerate}
    \end{thm}
    \begin{proof}
        For (1), in $M_{\leq d}$ we pivot non-square-free columns via rows coming from $\mathcal{F}_\text{Bool}$, then we perform the auxiliary Gaussian elimination step that produces $M_{\leq d}^\text{Bool}$ as submatrix, and finally we perform Gaussian elimination on $M_{\leq d}^\text{Bool}$.
        Since Gaussian elimination does not depend on the order of reductions $\rowspace \left( M_{\leq d}^\text{Bool} \right) \cup \mathcal{F}_\text{Bool}$ has to contain the $>$-Gr\"obner basis.

        For (2), it is well-known that in degree $i$ there exist $\binom{n}{i}$ square-free monomials in $P$.
        The polynomial $f \in \mathcal{F}$ contributes $\sum_{i = 0}^{d - \degree{f}} \binom{n}{i}$ many rows to the $K$-Boolean Macaulay matrix.
        Then, the total number of rows in $M_{\leq d}^\text{Bool}$ is bounded by
        \[
        	\sum_{f \in \mathcal{F}} \sum_{i = 0}^{d - \degree{f}} \binom{n}{i}
        	\leq \sum_{f \in \mathcal{F}} \sum_{i = 0}^{d - D} \binom{n}{i}
        	= \abs{F} \cdot \sum_{i = 0}^{d - D} \binom{n}{i}.
        \]
        The number of columns is bounded by $\sum_{i = 0}^{d} \binom{n}{i}$, so we also have that $\rank \left( M_{\leq d}^\text{Bool} \right) \leq \sum_{i = 0}^{d} \binom{n}{i}$, and the claim follows from the complexity of Gaussian elimination \cite[\S 2.2]{Storjohann-Matrix}.

        For (3), if $d - D \geq n$, then there exists at least one polynomial $f \in \mathcal{F}$ which has to be multiplied with any square-free monomial $s \in P$.
        So, $f$ contributes $\sum_{i = 0}^{n} \binom{n}{i} = 2^n$ many rows.
        Obviously, any other $g \in \mathcal{F}$ contributes at most $\leq 2^n$ many rows.
        Hence, the total number of rows is bounded by $\abs{\mathcal{F}} \cdot 2^n$, and the number of columns as well as the rank are bounded by $2^n$.
        Again, the claim follows by \cite[\S 2.2]{Storjohann-Matrix}.
    \end{proof}
    \begin{rem}
        This theorem can be refined by replacing the binomial coefficients $\binom{n}{i}$ with the number of monomials in $P \setminus \LM_> \left( \mathcal{F}_\text{Bool} \right)$ of degree $i$.
    \end{rem}

    So far our analysis does not differ from the one over $\F_2$, however the binary field has a special feature.
    Division with respect to the field equations $\left( x_1^2 - x_1, \dots, x_n^2 - x_n \right)$ is ``for free'', i.e.\ its complexity is negligible compared to Gaussian elimination.
    For general $K$-Boolean polynomial systems this is not necessarily the case.
    The complexity of general polynomial division with respect to quadratic $K$-Boolean polynomial systems follows as corollary to \cite[Proposition~3.10]{Tenti-Overdetermined}.
    \begin{cor}\label[cor]{Cor: Boolean Macaulay matrix construction complexity}
        Let $K$ be a field, let $\mathcal{F}_\text{Bool} \subset P = K [x_1, \dots, x_n]$ be a quadratic polynomial system, and let $>$ be a degree compatible term order on $P$.
        Assume that $\mathcal{F}_\text{Bool}$ is $K$-Boolean with respect to $>$.
        \begin{enumerate}
        	\item Let $f \in P$, then computation of $r_f = f \mod \left( \mathcal{F}_\text{Bool} \right)$ with respect to $>$ requires
        	\[
        		\mathcal{O} \left( \frac{n^2 + 3 \cdot n + 2}{2} \cdot \binom{n + \degree{f}}{\degree{f}} \right)
        	\]
        	field operations.

        	\item Let $\mathcal{F} \subset P$, $D_{\min} = \min_{f \in \mathcal{F}} \degree{f}$ and $D_{\max} = \max_{f \in \mathcal{F}} \degree{f}$, then construction of $M_{\leq d}^\text{Bool}$ for $\mathcal{F} \cup \mathcal{F}_\text{Bool}$ requires
        	\[
        		\mathcal{O} \left( \abs{\mathcal{F}} \cdot \frac{n^2 + 3 \cdot n + 2}{2} \cdot \sum_{i = 0}^{d - D_{\min}} \binom{n}{i} \cdot \binom{n + D_{\max} + i}{D_{\max} + i} \right)
        	\]
        	field operations.
        \end{enumerate}
    \end{cor}
    \begin{proof}
        For (1), in $P$ there exist $\sum_{i = 0}^{\degree{f}} \binom{n + i - 1}{i} = \binom{n + \degree{f}}{\degree{f}}$ monomials of degree $\leq \degree{f}$ and $\binom{n + 1}{2} + \binom{n}{1} + 1$ monomials of degree $\leq 2$.
        Now the claim follows from \cite[Proposition~3.10]{Tenti-Overdetermined}.

        For (2), with (1) we can estimate the complexity of the construction as
        \begin{align*}
        	&\sum_{f \in \mathcal{F}} \sum_{i = 0}^{d - \degree{f}} \binom{n}{i} \cdot \frac{n^2 + 3 \cdot n + 2}{2} \cdot \binom{n + \degree{f} + i}{\degree{f} + i} \\
        	&\stackrel{(\ast)}{\leq} \frac{n^2 + 3 \cdot n + 2}{2} \cdot \sum_{f \in \mathcal{F}} \sum_{i = 0}^{d - D_{\min}} \binom{n}{i} \cdot \binom{n + D_{\max} + i}{D_{\max} + i} \\
        	&= \abs{\mathcal{F}} \cdot \frac{n^2 + 3 \cdot n + 2}{2} \cdot \sum_{i = 0}^{d - D_{\min}} \binom{n}{i} \cdot \binom{n + D_{\max} + i}{D_{\max} + i},
        \end{align*}
        where $(\ast)$ follows from $d - \degree{f} \leq d - D_{\min}$ and the fact that $\binom{n + d}{d}$ is strictly increasing in $d$, hence $\binom{n + \degree{f} + i}{\degree{f} + i} \leq \binom{n + D_{\max} + i}{D_{\max} + i}$.
    \end{proof}

    Of course, the $K$-Boolean Macaulay matrix generalizes to more advanced linear algebra-based Gr\"obner basis algorithms like F4/5 \cite{Faugere-F4,Faugere-F5}.

    \subsection{The Eigenvalue Method}\label{Sec: Eigenvalue Method}
    Now we discuss how the zeros of a polynomial system can be found via linear algebra once a zero-dimensional DRL Gr\"obner basis of the input system is known \cite[Chapter~6]{Kreuzer-CompLinAlg}.
    Let $I \subset P = K [x_1, \dots, x_n]$ be a zero-dimensional ideal, i.e.\ $P / I$ is a finite dimensional $K$-vector space \cite[Proposition~3.7.1]{Kreuzer-CompAlg1}, and let $d = \dim_K \left( P / I \right)$.
    If $>$ is a term order on $P$ and $\mathcal{G} \subset I$ is $>$-Gr\"obner basis, then a $K$-vector space basis $\mathcal{B}$ of $P / I$ is given by the monomials not contained in $\LM_> \left( \mathcal{G} \right) = \LM_> \left( I \right)$.
    Now let $f \in P$, then the multiplication map
    \begin{equation}
    	\theta_f: P / I \to P / I, \qquad\qquad x \mapsto x \cdot f,
    \end{equation}
    is $K$-linear, so it can be represented by a matrix.
    This matrix is called the multiplication matrix $\mathbf{M}_f$ for $f$.
    Given a $>$-Gr\"obner basis its computation is straight-forward: Index the columns of $\mathbf{M}_f$ by the $K$-vector space basis $\mathcal{B}$, and the rows by $b \cdot f$, where $b \in \mathcal{B}$.
    The row $b \cdot f$ is then simply the coefficient vector of $b \cdot f$ in $P / I$, which can be computed via $b \cdot f \mod (\mathcal{G})$ with respect to $>$.

    Now assume in addition that $K$ is algebraically closed, then for every point in the variety $\mathbf{x} \in \mathcal{V} \left( I \right) = \left\{ \mathbf{x} \in K^n \mid \forall g \in I \!: g (\mathbf{x}) = 0 \right\}$ the point $f (\mathbf{x})$ is an eigenvalue of $\mathbf{M}_f$, see \cite[Proposition~6.2.1]{Kreuzer-CompLinAlg}.
    Therefore, we can compute the variety $\mathcal{V} \left( I \right)$ as follows: Compute the eigenvalues of $\mathbf{M}_{x_1}, \dots, \mathbf{M}_{x_n}$, iterate over all possible combinations, and verify whether a combination is indeed a solution to $I$.
    This approach is known as \emph{Eigenvalue Method} \cite[\S 6.2.A]{Kreuzer-CompLinAlg}.

	It is well-known that the eigenvalues of a matrix $\mathbf{M} \in K^{N \times N}$ are given by the roots of its characteristic polynomial
	\begin{equation}
		\chi_\mathbf{M} (x) = \det \left( \mathbf{I}_{N \times N} \cdot x - \mathbf{M} \right) \in K [x].
	\end{equation}
	So the complexity of the Eigenvalue Method boils down to $n$ determinant computations with polynomial entries.

    \subsubsection{Solving Structured Polynomial Systems Over Finite Fields}\label{Sec: solving structured systems}
    Let $\mathcal{F} = \{ f_1, \dots, f_n \} \subset P = K [x_1, \dots, x_n]$ be such that $\LM_{DRL} \left( f_i \right) = x_i^{d_i}$ for all $1 \leq i \leq n$.
    Then, trivially $\mathcal{F}$ is a DRL Gr\"obner basis, see e.g.\ \cite[Chapter~2~\S 9~Theorem~3, Proposition~4]{Cox-Ideals} with $\dim_K \left( \mathcal{F} \right) = \prod_{i = 1}^{n} d_i$.
    Let $D = \prod_{i = 1}^{n - 1} d_i$, and let $\mathcal{B}' \subset K [x_1, \dots, x_{n - 1}] / \left( x_1^{d_1}, \dots, x_{n - 1}^{d_{n - 1}} \right)$ be some $K$-vector space basis, then $\mathcal{B} = \bigcup_{i = 0}^{d_n - 1} x_n^i \cdot \mathcal{B}'$ is a $K$-vector space basis for $P / (\mathcal{F})$.
    In particular, the multiplication matrix $\mathbf{M}_{x_n}$ is then of the form
    \begin{equation}\label{Equ: multiplication matrix}
        \mathbf{M}_{x_n} =
        \begin{blockarray}{ c c c c c c }
            \mathcal{B}' & x_{n} \cdot \mathcal{B}' & x_{n}^{2} \cdot \mathcal{B}' & \ldots &  x_{n}^{d_{n - 1}} \cdot \mathcal{B}' \\
            \begin{block}{ ( c c c c c ) c}
                \mathbf{0}     & \mathbf{I}     & \mathbf{0} & \ldots & \mathbf{0}             & x_{n} \cdot \mathcal{B}'             \\
                \mathbf{0}     & \mathbf{0}     & \mathbf{I} & \ldots & \mathbf{0}             & x_{n}^{2} \cdot \mathcal{B}'         \\
                \vdots         & \vdots         & \vdots     & \ddots & \vdots                 & \vdots                               \\
                \mathbf{0}     & \mathbf{0}     & \mathbf{0} & \ldots & \mathbf{I}             & x_{n}^{d_{n} - 1} \cdot \mathcal{B}' \\
                \mathbf{A}_{0} & \mathbf{A}_{1} & \ldots     & \ldots & \mathbf{A}_{d_{n} - 1} & x_{n}^{d_{n}} \cdot \mathcal{B}'     \\
            \end{block}
        \end{blockarray}
        ,
    \end{equation}
    where $\mathbf{0} \in K^{D \times D}$ denotes the zero matrix, $\mathbf{I} \in K^{D \times D}$ the identity matrix, and $\mathbf{A}_{0}, \dots, \allowbreak \mathbf{A}_{d_{n} - 1} \in K^{D \times D}$ some matrices.
    By \cite[Lemma~2]{C:BBLAOP24} we then have that
    \begin{equation}\label{Equ: characteristic polynomial}
        \chi_{\mathbf{M}_{x_n}} (x) = \det \left( x_n^{d_n} \cdot \mathbf{I} - \sum_{i = 0}^{d_{n} - 1} x_n^i \cdot \mathbf{A}_i \right),
    \end{equation}
    and computation of this determinant requires \cite[\S 3.2]{C:BBLAOP24}
    \begin{equation}\label{Equ: characteristic polynomial complexity}
        \mathcal{O} \bigg( d_{n} \cdot \log_{2} (d_{n})^{2} \cdot \Big( 1 + \log_{2} \big( \log_{2} (d_{n}) \big) \Big) \cdot \Big( \frac{D}{d_{n}} \Big)^{\omega} \bigg)
    \end{equation}
    field operations, where $2 \leq \omega < 2.371552$ is again the linear algebra constant \cite{SODA:WXXZ24}.

    In this paper we solely work over finite fields $\Fq$, and we are only interested in the $\Fq$-valued roots of $\mathcal{F}$.
    Given a univariate polynomial $f \in \Fq$ with $\degree{f} = d$ we can always achieve this restriction by passing to the GCD with the field equation $g = \gcd \left( f, x^q - x \right)$.
    By \cite[\S 3.1]{ToSC:BBLP22} computation of this GCD requires
    \begin{equation}\label{Equ: GCD with field equation complexity}
        \mathcal{O} \left(
        \begin{rcases}
            \begin{dcases}
                d \cdot \log_2 (d) \cdot \log \big( \log (d) \big) \cdot \big( \log (d) + \log_2 (q) \big), & d < q, \\
                d \cdot \log (d)^2 \cdot \log \big( \log (d) \big), & d \geq q
            \end{dcases}
        \end{rcases}
        \right)
    \end{equation}
    field operations.
    Provided that only a few roots of $f$ come from $\Fq$ the final factoring step can be neglected.

    In this paper we will encounter quadratic DRL Gr\"obner bases $\mathcal{F} = \{ f_1, \dots, f_n \} \subset  P = \Fq [x_1, \dots, x_n]$ of the form
    \begin{equation}\label{Equ: special shape}
        f_i = \mathcal{Q} \left( x_i, \dots, x_n \right) + \mathcal{A}_i \left( x_1, \dots, x_n \right),
    \end{equation}
    where $\mathcal{Q}_i \in P$ is a quadratic polynomial which has the term $x_i^2$ and $\mathcal{A}_i \in P$ is an affine polynomial which has the term $x_{n - i + 1}$ for $1 \leq i \leq n$.
    It is easy to see that $\mathcal{F}$ is a DRL Gr\"obner basis for $x_1 > \ldots > x_n$.
    For such systems we can give an optimized Eigenvalue Method to find the $\Fq$-valued solutions.
    We use a $\Fq$-vector space basis such that the multiplication matrix $\mathbf{M}_{x_n}$ has the shape from \Cref{Equ: multiplication matrix}.
    Then, the characteristic polynomial is constructed in $\mathcal{O} \left( 2^{\omega \cdot (n - 1) + 1} \right)$ field operations, and its $\Fq$-valued roots are extracted via the GCD with the field equation.
    Let $\alpha \in \Fq$ be such a solution for $x_n$, then we can substitute $x_n = \alpha$ back into $\mathcal{F}$.
    In particular, from $f_n$ we obtain the affine equation
    \begin{equation}
        \alpha^2 + \tilde{\mathcal{A}_n} = 0.
    \end{equation}
    Now we rearrange this equation to have $x_{1}$ on the right-hand side.
    Then, we can additionally eliminate the variable $x_{1}$ in the equations $f_{1} = \ldots = f_{n - 1} = 0$.
    Since $x_1$ is not present in the quadratic terms of $f_2, \dots, f_{n - 1} \in \Fq [x_2, \dots, x_{n - 1}]$, they are again a zero-dimensional DRL Gr\"obner basis.
    The polynomial $f_1$ is an additional overdetermined polynomial which we ignore.
    Now we iterate this procedure to solve for $x_{n - 1}, x_{n - 2}, \dots, x_{\floor{{\frac{n}{2}}}}$, then we lift to a global solution guess for $x_1, \dots, x_n$, and finally verify it on $\mathcal{F}$.
    Let $\mathcal{D} (d, q)$ denote the complexity of the GCD between a degree $d$ polynomial and the field equation of $\Fq$, moreover let $N$ be a bound on the number of $\Fq$-valued solutions at each iteration in our solving algorithm.
    Then, this dedicated Eigenvalue Method requires
    \begin{align}
        &\mathcal{O} \left(
        \begin{rcases}
            \begin{dcases}
                \sum_{i = 1}^{\frac{n}{2}} N^{\frac{n}{2} - i} \cdot \Big( 2^{\omega \cdot (2 \cdot i - 1) + 1} + \mathcal{D} \big( 2^{2 \cdot i}, q \big) \Big), & 2 \mid n, \\
                \sum_{i = 1}^{\frac{n - 1}{2}} N^{\frac{n - 1}{2} - i} \cdot \Big( 2^{\omega \cdot 2 \cdot i + 1} + \mathcal{D} \big( 2^{2 \cdot i + 1}, q \big) \Big), & 2 \nmid n
            \end{dcases}
        \end{rcases}
        \right) \\
        &= \mathcal{O} \left(
        \begin{rcases}
            \begin{dcases}
                2^{\omega + 1} \cdot \frac{2^{\omega \cdot n} - N^\frac{n}{2}}{2^{2 \cdot \omega} - N} + \sum_{i = 1}^{\frac{n}{2}} N^{\frac{n}{2} - i} \cdot  \mathcal{D} \big( 2^{2 \cdot i}, q \big), & 2 \mid n, \\
                2^{2 \cdot \omega + 1} \cdot \frac{2^{\omega \cdot (n - 1)} - N^\frac{n - 1}{2}}{2^{2 \cdot \omega} - N} + \sum_{i = 1}^{\frac{n - 1}{2}} N^{\frac{n - 1}{2} - i} \cdot \mathcal{D} \big( 2^{2 \cdot i + 1}, q \big), & 2 \nmid n
            \end{dcases}
        \end{rcases}
        \right) \label{Equ: complexity solving structured systems}
    \end{align}
    field operations.

    \section{Analysis of \Ciminion}\label{Sec: Analysis of Ciminion}
    At round level \Ciminion is a three branch Feistel network.
    In \cite[\S 4.1]{ToSC:Steiner24} a DRL Gr\"obner basis for Feistel-$2n/n$ polynomial systems has been computed by substituting the linear branches into the non-linear ones.
    In principle, we are going to replicate this strategy to construct a DRL Gr\"obner basis for \Ciminion.

    \subsection{A DRL Gr\"obner Basis for \Ciminion}\label{Sec: Ciminion Groebner basis}
    For simplicity let us denote the matrix in the $i$\textsuperscript{th} \Ciminion round by $\mathbf{A}_i \in \Fq^{3 \times 3}$, see \Cref{Equ: Ciminion round function}, then
    \begin{equation}\label{Equ: inverse Ciminion matrix}
        \mathbf{A}_i =
        \begin{pmatrix}
            0 & 0         & 1         \\
            1 & c_4^{(i)} & c_4^{(i)} \\
            0 & 1         & 1
        \end{pmatrix}
        , \qquad
        \mathbf{A}_i^{-1} =
        \begin{pmatrix}
            0  & 1 & -c_4^{(i)} \\
            -1 & 0 & 1          \\
            1  & 0 & 0
        \end{pmatrix}
        ,
    \end{equation}
    where $c_4^{(i)} \in \Fq \setminus \{ 0, 1 \}$.
    \begin{thm}\label{Th: Ciminion Groebner basis}
        Let $\Fq$ be a finite field, and let $\mathcal{F}_\Ciminion = \left\{ \mathbf{f}^{(1)}, \dots, \mathbf{f}^{(r)} \right\} \subset \Fq \Big[ y_1, y_2, \mathbf{x}^{(1)}, \allowbreak \dots, \mathbf{x}^{(r - 1)}, x \Big]$ be a \Ciminion polynomial system.
        For the DRL term order $y_1 > y_2 > \mathbf{x}^{(1)} > \ldots > \mathbf{x}^{(r - 1)} > x$, the following procedure computes a DRL Gr\"obner basis for \Ciminion:
        \begin{enumerate}
            \item For every round $1 \leq i \leq r$ compute $\mathbf{g}^{(i)} = \mathbf{A}_i^{-1} \mathbf{f}_i$.

            \item Collect the linear polynomials from the $\mathbf{g}^{(i)}$'s in $\mathcal{L}$, and perform Gaussian elimination on $\mathcal{L}$.

            \item Reduce the non-linear polynomials in the $\mathbf{g}^{(i)}$'s modulo $\mathcal{L}$.
        \end{enumerate}
    \end{thm}
    \begin{proof}
        By inspection of the round function (\Cref{Equ: Ciminion round function}) we have that
        \begin{align}
            \mathbf{g}^{(1)} &=
            \begin{pmatrix}
                \aleph \\
                y_1 \\
                \aleph \cdot y_1 + y_2
            \end{pmatrix}
            + \mathbf{A}_1^{-1} \left( \mathbf{c}^{(1)} - \mathbf{x}^{(1)} \right), \label{Equ: first round equation matrix inverted}\\
            \mathbf{g}^{(i)} &=
            \begin{pmatrix}
                x_{1}^{(i - 1)} \\
                x_{2}^{(i - 1)} \\
                x_{3}^{(i - 1)} + x_{1}^{(i - 1)} \cdot x_{2}^{(i - 1)}
            \end{pmatrix}
            + \mathbf{A}_i^{-1} \left( \mathbf{c}^{(i)} - \mathbf{x}^{(i)} \right), \label{Equ: round equation matrix inverted} \\
            \mathbf{g}^{(r)} &=
            \begin{pmatrix}
                x_{1}^{(r - 1)} \\
                x_{2}^{(r - 1)} \\
                x_{3}^{(r - 1)} + x_{1}^{(r - 1)} \cdot x_{2}^{(r - 1)}
            \end{pmatrix}
            + \mathbf{A}_{r}^{-1} \left( \mathbf{c}^{(r)} -
            \begin{pmatrix}
                c_1 - p_1 \\
                c_2 - p_2 \\
                x
            \end{pmatrix}
            \right). \label{Equ: last round equation matrix inverted}
        \end{align}
        Hence, $\mathcal{L}$ contains $\mathbf{g}^{(1)}$ and the first two components of all remaining $\mathbf{g}^{(i)}$'s.

        For Gaussian elimination on $\mathcal{L}$ let us start in the last round $i = r$.
        By \Cref{Equ: inverse Ciminion matrix,Equ: last round equation matrix inverted} the linear polynomials with leading monomials $x_{1}^{(r - 1)}$, $x_{2}^{(r - 1)}$ contain a linear term in $x$ and a constant term.
        Moving to $i = r - 1$ next the linear polynomials from $\mathbf{g}^{(r - 1)}$ depend on $\mathbf{x}^{(r - 2)}$ and $\mathbf{x}^{(r - 1)}$, and via Gaussian elimination the variables $x_{1}^{(r - 1)}$ and $x_{2}^{(r - 1)}$ get eliminated.
        Thus, the linear polynomials with leading monomials in $x_{1}^{(r - 2)}$ and $x_{2}^{(r - 2)}$ can only depend on $x_{3}^{(r - 1)}$ and $x$.
        Moreover, by \Cref{Equ: inverse Ciminion matrix} the coefficients for $x_{3}^{(r - 1)}$ are non-zero.
        Via a downwards induction we can then conclude for $2 \leq i \leq r - 1$ that after Gaussian elimination
        \begin{align*}
        	\LM_{DRL} \left( g_1^{(i)} \right) &= x_{1}^{(i - 1)}, \\
        	\LM_{DRL} \left( g_2^{(i)} \right) &= x_{2}^{(i - 1)}, \\
        	g_1^{(i)}, g_2^{(i)} &\in \Fq \big[ x_{3}^{(i)}, \dots, x_{3}^{(r - 1)}, x \big],
        \end{align*}
        and that the coefficients of $x_{3}^{(i)}$ in $g_1^{(i)}$ and $g_2^{(i)}$ are non-zero.

        The first round needs special care since the first row of $\mathbf{A}_1^{-1}$ is non-zero on the variable $x_{2}^{(1)}$.
        (If it were zero on $x_{1}^{(1)}$ and $x_{2}^{(1)}$, the first polynomial in $\mathbf{g}^{(1)}$ would have leading monomial $x_{3}^{(1)}$, and we could proceed.)
        Expansion of \Cref{Equ: first round equation matrix inverted} yields
        \[
            \mathbf{g}^{(1)} =
            \begin{pmatrix}
                \aleph \\
                y_1 \\
                \aleph \cdot y_1 + y_2
            \end{pmatrix}
            + \mathbf{A}_1^{-1} \mathbf{c}^{(1)} -
            \begin{pmatrix}
                x_{2}^{(1)} - c_4^{(1)} \cdot x_{3}^{(1)} \\
                -x_{1}^{(1)} + x_{3}^{(1)} \\
                x_{1}^{(1)}
            \end{pmatrix}
            .
        \]
        Producing a polynomial with leading monomial $y_2$ is straight-forward, we simply have to take $g_3^{(1)} - \aleph \cdot g_2^{(1)}$.
        On the other hand, there is second polynomial with leading monomial $x_{2}^{(1)}$ after inversion of the matrices
        \[
            g_2^{(2)} = x_{2}^{(1)} + \left( \mathbf{A}_2^{-1} \mathbf{c}^{(2)} \right)_2 - \left( -x_{1}^{(2)} + x_{3}^{(2)} \right).
        \]
        Then
        \[
            g_1^{(1)} + g_2^{(2)} = \aleph + \left( \mathbf{A}_1^{-1} \mathbf{c}^{(1)} \right)_1 + c_4^{(1)} \cdot x_{3}^{(1)} + \left( \mathbf{A}_2^{-1} \mathbf{c}^{(2)} \right)_2 + x_{1}^{(2)} - x_{3}^{(2)},
        \]
        and by the choice of our DRL term order and $c_4^{(1)} \neq 0$ this polynomial has leading monomial $x_{3}^{(1)}$.
        Obviously, after Gaussian elimination of $\mathcal{L}$ the terms of $g_1^{(1)} + g_2^{(2)}, g_2^{(1)}, g_3^{(1)} - \aleph \cdot g_2^{(1)}$ after the leading monomials will also only depend on the variables $x_{3}^{(2)}, \dots, x_{3}^{(r - 1)}, x$.

        Lastly, let us consider the reduction modulo $\mathcal{L}$.
        By our previous observations all linear equations are of the form
        \begin{align}
            y_1 + \mathcal{A}_1 \left( x_{3}^{(2)}, \dots, x_{3}^{(r - 1)}, x \right) &= 0, \nonumber \\
            y_2 +\mathcal{A}_2 \left( x_{3}^{(2)}, \dots, x_{3}^{(r - 1)}, x \right) &= 0, \nonumber \\
            x_{3}^{(1)} + \mathcal{A}_3 \left( x_{3}^{(2)}, \dots, x_{3}^{(r - 1)}, x \right) &= 0, \nonumber \\
            x_{1}^{(i - 1)} + c_4^{(i)} \cdot x_{3}^{(i)} + \ldots &= 0, \nonumber \\
            x_{2}^{(i - 1)} - x_{3}^{(i)} + \ldots &= 0, \nonumber \\
            x_{1}^{(r - 1)} + c_4^{(r)} \cdot x + \ldots &= 0, \nonumber \\
            x_{2}^{(r - 1)} - x + \ldots &= 0, \nonumber
        \end{align}
        where $\mathcal{A}_1$, $\mathcal{A}_2$ and $\mathcal{A}_3$ are affine polynomials and $2 \leq i \leq r - 1$.
        So, for $2 \leq i \leq r$ we can directly conclude from \Cref{Equ: round equation matrix inverted,Equ: last round equation matrix inverted} that
        \[
            \hat{g}_3^{(i)} = \left( g_3^{(i)} \mod (\mathcal{L}) \right) \in \Fq \Big[ x_{3}^{(2)}, \dots, x_{3}^{(r - 1)}, x \Big],
        \]
        has leading monomial ${x_{3}^{(i)}}^2$ and respectively $x^2$ for the last round.

        Thus, we have constructed a polynomial system with pairwise coprime leading monomials, so by \cite[Chapter 2 §9 Theorem 3, Proposition 4]{Cox-Ideals} we have constructed a Gr\"obner basis.
    \end{proof}

    Note that if we only keep the $r - 1$ quadratic polynomials of $\mathcal{G}$, then we still have a fully determined polynomial system in $r - 1$ variables.
    Therefore, we call the quadratic polynomials of $\mathcal{G}$ the downsized \Ciminion polynomial system.

    For the next corollary, recall that a LEX Gr\"obner basis $\mathcal{F} = \{ f_1, \dots, f_n \} \subset K [x_1, \dots, \allowbreak x_n]$ is said to be in $x_n$-\emph{shape position} if it is of the form
    \begin{align}
        f_i &= x_i - \tilde{f}_n (x_n), \qquad 1 \leq i \leq n - 1, \\
        f_n &= \tilde{f}_n (x_n)
    \end{align}
    It is well-known that the LEX Gr\"obner basis of a radical ideal can be brought into shape position for some ordering of the variables, see \cite[Theore~3.7.25]{Kreuzer-CompAlg1}.
    \begin{cor}\label[cor]{Cor: Ciminion Macaualy bound}
        In the situation of \Cref{Th: Ciminion Groebner basis}, let $\mathcal{G}_\Ciminion$ denote the \Ciminion DRL Gr\"obner basis.
        \begin{enumerate}
            \item $\dim_{\Fq} \left( \mathcal{F}_\Ciminion \right) = 2^{r - 1} = 2^{r_C + r_E - 1}$.

            \item Let $I \subset \overline{\Fq} \Big[ y_1, y_2, \mathbf{x}^{(1)}, \dots, \mathbf{x}^{(r - 1)}, x \Big] [x_0]$ be a homogeneous ideal such that $\mathcal{Z}_+ (I) \neq \emptyset$ and $\mathcal{G}_{\Ciminion}^{\homog} \subset I$.
            Then $I$ is in generic coordinates.

            \item $\solvdeg_{DRL} \left( \mathcal{G}_\Ciminion \right) \leq r = r_C + r_E$.

            \item The LEX Gr\"obner basis for \Ciminion under $y_1 > y_2 > \mathbf{x}^{(1)} > \ldots > \mathbf{x}^{(r - 1)} > x$ is in $x$-shape position.
        \end{enumerate}
    \end{cor}
    \begin{proof}
        (1), (2) and (3) are consequences of the DRL Gr\"obner basis from \Cref{Th: Ciminion Groebner basis} and application of \cite[Theorem~3.2]{ToSC:Steiner24} and \Cref{Th: solvdeg and CM-regularity}.
        For the last claim we take a look at the shape of the quadratic polynomials in the DRL Gr\"obner basis
        \begin{align*}
            g_3^{(2)} &= \mathcal{Q}_2 \left( x_{3}^{(2)}, \dots, x_{3}^{(r - 1)}, x \right) + \mathcal{A}_2 \left( x_{3}^{(2)}, \dots, x_3^{(r - 1)}, x \right), \\
            g_3^{(i)} &= \mathcal{Q}_i \left( x_{3}^{(i)}, \dots, x_{3}^{(r - 1)}, x \right) + \mathcal{A}_i \left( x_{3}^{(i - 1)}, \dots, x_{3}^{(r - 1)}, x \right),  \qquad 3 \leq i \leq r \\
            g_3^{(r)} &= \mathcal{Q} \left( x \right) + \mathcal{A}_{r} \left( x_3^{(r - 1)}, x \right),
        \end{align*}
        where the $\mathcal{Q}_i$'s are quadratic polynomials which contain the monomial ${x_{3}^{(i)}}^2$ and respectively $x^2$ for the last round, and the $\mathcal{A}_i$'s are affine polynomials that contain the monomial $x_{3}^{(i - 1)}$ for all $i \geq 3$ and without any restrictions on $\mathcal{A}_2$.
        Then, by iterative substitution of $g_3^{(r)}$ into $g_3^{(r - 1)}$ to eliminate $x_3^{(r - 1)}$, $g_3^{(r - 1)}$ into $g_3^{(r - 2)}$ to eliminate $x_3^{(r - 2)}$, etc.\  we produce a LEX Gr\"obner basis in $x$-shape position.
    \end{proof}

    \subsection{Bariant's Attack}\label{Sec: Bariant's attack}
    Recently Bariant introduced an attack on \Ciminion via univariate greatest common divisor computations \cite{Bariant-Ciminion}.
    As in our polynomial model, he introduced an auxiliary variable $X$ for the truncated component.
    If one inverts $p_C \circ p_E$ using $(c_1 - p_1, c_2 - p_2, X)^\intercal$ as output state, then one obtains three polynomials, see \cite[Fig.~4]{Bariant-Ciminion},
    \begin{equation}\label{Equ: Bariant polynomials}
        \aleph = f_1 (X), \qquad K_1 = f_2 (X), \qquad K_2 = f_3 (X).
    \end{equation}
    Since the first polynomial is univariate in $X$, one can compute its greatest common divisor with $X^q - X$ to solve for $X$ and henceforth also solve for $K_1$ and $K_2$.

    Observe that Bariant's univariate polynomial is exactly the univariate polynomial of the LEX Gr\"obner basis.
    \begin{lem}\label[lem]{Lem: Bariant polynomial}
        Let $\Fq$ be a finite field, let $I_\Ciminion \subset \Fq \big[ y_1, y_2, \mathbf{x}^{(1)}, \ldots, \mathbf{x}^{(r - 1)}, x \big]$ be a \Ciminion ideal, let $\tilde{f} \in \Fq [x]$ be the univariate LEX polynomial from \Cref{Cor: Ciminion Macaualy bound}, and let $f \in \Fq [x]$ be Bariant's univariate polynomial.
        Then $f = \alpha \cdot \tilde{f}$ for some $\alpha \in \Fqx$.
    \end{lem}
    \begin{proof}
        \Ciminion is an iterated Feistel permutation therefore we have that
        \[
            \mathcal{R}^{(i)} \left( \mathbf{x}^{(i - 1)} \right) \equiv \mathbf{x}^{(i)} \mod I_\Ciminion \qquad \Longleftrightarrow \qquad \mathbf{x}^{(i - 1)} \equiv {\mathcal{R}^{(i)}}^{-1} \left( \mathbf{x}^{(i)} \right) \mod I_\Ciminion,
        \]
        where $2 \leq i \leq r - 1$, and analogously for $i = 1, r$.
        Since the inverse round functions are in $I_\Ciminion$, we can iteratively substitute them into each other.
        This yields exactly the polynomials from \Cref{Equ: Bariant polynomials}, so $f \in I_\Ciminion$.
        By \Cref{Cor: Ciminion Macaualy bound} we can construct a \Ciminion LEX Gr\"obner basis in $x$-shape position, then for any $g \in \Fq [x]$ we have $g \in I_\Ciminion \Leftrightarrow \tilde{f} \mid g$ \cite[Chapter~4~\S 5~Exercise~13]{Cox-Ideals}.
        Since $\degree{\tilde{f}} = \degree{f}$ they can only differ up to a multiplicative constant.
    \end{proof}

    Moreover, the construction of the univariate LEX polynomial can be effectively expressed in the cost of univariate polynomial multiplication over finite fields \cite{Cantor-Multiplication}, which is quasi-linear in the degree.

    As a countermeasure, Bariant proposes to add a linear combination of $K_1$ and $K_2$ to each component after $p_C$ \cite[\S 4]{Bariant-Ciminion}.
    We note that in general this modification breaks our Gr\"obner basis from \Cref{Th: Ciminion Groebner basis}.
    As a settlement, we propose to add a linear combination of $K_1$ and $K_2$ only to the non-linear parts of the round functions, i.e.\
    \begin{equation}\label{Equ: Ciminion fix}
        \begin{split}
            \mathcal{R}^{(i)}: \Fq^3 \times \Fq^2 &\to \Fq^3, \\
            \begin{pmatrix}
                x \\ y \\ z
            \end{pmatrix}
            ,
            \begin{pmatrix}
                K_1 \\ K_2
            \end{pmatrix}
            &\mapsto
            \begin{pmatrix}
                0 & 0         & 1         \\
                1 & c_4^{(i)} & c_4^{(i)} \\
                0 & 0         & 1
            \end{pmatrix}
            \begin{pmatrix}
                x \\ y \\ z + x \cdot y + \alpha \cdot K_1 + \beta \cdot K_2
            \end{pmatrix}
            + \mathbf{c}^{(i)},
        \end{split}
    \end{equation}
    where $\alpha, \beta \in \Fqx$.
    This modification does not affect the proof of \Cref{Th: Ciminion Groebner basis}, hence we maintain the DRL Gr\"obner basis for the modified \Ciminion polynomial system.
    However, this modification breaks the construction of the univariate LEX polynomial via univariate polynomial multiplication, because the quadratic DRL Gr\"obner basis polynomial for the $r_C$\textsuperscript{th} round will now also contain linear terms for the variables $x_{3}^{(2)}, \dots, x_{3}^{(r - 1)}, x$ instead of only $x_{3}^{(r_C - 1)}, \dots, x_{3}^{(r - 1)}, x$.
    While the LEX Gr\"obner basis of the modified \Ciminion might still be in $x$-shape position, we can now construct it only via an FGLM algorithm.

    Note if the modification of \Cref{Equ: Ciminion fix} is performed in all rounds of $p_E$, then the \emph{Middle State-Output Relation} polynomial system of Bariant et al.\ \cite[\S 5]{ToSC:BBLP22} becomes underdetermined, so this attack also becomes more difficult.

    \subsection{Implementation}
    We implemented \Ciminion in the computer algebra systems \SageMath \cite{SageMath} and \OSCAR \cite{OSCAR}.\footnote{\repository}
    The \SageMath implementation can handle $128$ bit prime numbers, in particular the Gr\"obner basis from \Cref{Th: Ciminion Groebner basis} can be computed for primes of that size.
    For $p = 2^{127} + 45$, $r_C = 90$ and $r_E = 14$ it is computed in negligible time.

    \subsubsection{\Ciminiontwo}\label{Sec: Ciminion2}
    We also implemented a modified version of \Ciminion to protect against Bariant's attack, called \Ciminiontwo, which adds $K_1 + K_2$ in every non-linear round of $p_E$ and $p_C$, i.e.\ only in the first round of $p_E$ the key addition is omitted.
    The round function for the first sample then modifies to
    \begin{equation}
        \mathcal{R}^{(i)}:
        \begin{pmatrix}
            x \\ y \\ z
        \end{pmatrix}
        ,
        \begin{pmatrix}
            K_1 \\ K_2
        \end{pmatrix}
        \mapsto
        \begin{pmatrix}
            0 & 0         & 1         \\
            1 & c_4^{(i)} & c_4^{(i)} \\
            0 & 1         & 1
        \end{pmatrix}
        \begin{pmatrix}
            x \\ y \\ z + x \cdot y + K_1 + K_2
        \end{pmatrix}
        + \mathbf{c}^{(i)},
    \end{equation}
    where $2 \leq i \leq r_C + r_E$.
    For higher samples one then adds $K_{2 \cdot j - 1} + K_{2 \cdot j}$ in $p_{E}$.
    For \Ciminiontwo a DRL Gr\"obner basis for the first sample is also computed in negligible time in \SageMath for practical parameters.
    Due to the additional key addition the quadratic polynomials in the DRL Gr\"obner basis modify to
    \begin{align}
        g_3^{(i)} &= \mathcal{Q}_i \left( x_3^{(i)}, \dots, x_3^{(r - 1)}, x \right) + \mathcal{A}_i \left( x_3^{(2)}, \dots, x_3^{(r - 1)}, x \right), \qquad 2 \leq i \leq r - 1, \\
        g_3^{(r)} &= \mathcal{Q}_r (x) + \mathcal{A}_i \left( x_3^{(2)}, \dots, x_3^{(r - 1)}, x \right),
    \end{align}
    where $\mathcal{Q}_i$ is quadratic and has the term ${x_3^{(i)}}^2$ and respectively $x^2$ for the last round, and $\mathcal{A}_i$ is affine and has non-zero coefficients for $x_3^{(2)}, \dots, x_3^{(r - 1)}, x$.

    \subsection{Cryptanalysis}\label{Sec: Ciminion cryptanalysis}
    We now analyze the complexity of a Gr\"obner basis attack on \Ciminion and \Ciminiontwo via our Gr\"obner basis, and compare it to the designer's estimation and Bariant's attack.

    \subsubsection{Estimation of the Designers}
    For Gr\"obner basis cryptanalysis of \Ciminion \cite[\S 4.4]{EC:DGGK21} the designers considered a weaker scheme where the keys are added after the permutations $p_E$ instead of before the rolling functions, see \cite[Fig.~7]{EC:DGGK21}.
    The designers claimed that the weaker scheme is simpler to analyze than \Ciminion.
    We reject this claim, after all we have computed a DRL Gr\"obner basis for \Ciminion with rather simple linear variable eliminations, see \Cref{Th: Ciminion Groebner basis}.

    The designers assumed that weakened \Ciminion polynomial systems are regular, moreover for regular systems it is well-known that the degree of regularity $d_{\reg}$ is equal to the Macaulay bound (\Cref{Th: solvdeg and CM-regularity}).
    Then, the DRL Gr\"obner basis complexity is bounded by \cite[Theorem~7]{Bardet-Complexity}
    \begin{equation}\label{Equ: complexity regular systems}
    	\mathcal{O} \left( \binom{n + d_{\reg}}{d_{\reg}}^\omega \right)
    \end{equation}
    field operations.
    For \Ciminion cryptanalysis, it is then conjectured that the DRL Gr\"obner basis computation is the most costly part in an attack.

    The designers did not analyze the iterated model, they only studied a model of two equations in two key variables, we call this model the \emph{fully substituted model}.
    By construction the polynomials in the fully substituted model have degree $2^{r_C + r_E - 1}$.
    Under the regularity assumption we then have that $d_{\reg} = 2 \cdot \left( 2^{r_C + r_E - 1} - 1 \right) + 1 = 2^{r_C + r_E} - 1$, so \Cref{Equ: complexity regular systems} evaluates to
    \begin{equation}
        \mathcal{O} \left( \binom{2^{r_C + r_E} + 1}{2^{r_C + r_E} - 1}^{\omega} \right)
        = \mathcal{O} \left( \left( \frac{2^{r_C + r_E} \cdot \left( 2^{r_C + r_E} + 1 \right)}{2} \right)^\omega \right).
    \end{equation}
    In \Cref{Tab: Ciminion complexity} this estimate is evaluated for sample we round numbers.

    \subsubsection{Complexity of Bariant's Attack}\label{Sec: Bariant attack complexity}
    Bariant's attack is split into two parts: construction of the univariate polynomial via an iterated sequence of squarings and the GCD with the field equation.
    Multiplication of two univariate polynomials $f, g \in \Fq [x]$ with $\degree{f}, \degree{g} \leq d$ can be performed in $\mathcal{O} \Big( d \cdot \log_2 (d) \cdot \log_2 \big( \log_2 (d) \big) \Big)$ field operations \cite{Cantor-Multiplication}, and the complexity for the GCD was given in \Cref{Equ: GCD with field equation complexity}.
    The construction step then requires \cite[\S 3.2]{Bariant-Ciminion}
    \begin{equation}
        \mathcal{O} \big( 2^{r_C + r_E} \cdot (r_C + r_E) \cdot \log_2 (r_C + r_E) \big)
    \end{equation}
    field operations, and the GCD requires
    \begin{equation}
        \mathcal{O} \Big( 2^{r_C + r_E - 1} \cdot (r_C + r_E - 1) \cdot \big( r_C + r_E - 1 + \log_2 (q) \big) \cdot \log_2 (r_C + r_E) \Big).
    \end{equation}
    Therefore, in general the GCD step dominates the complexity.
    In \Cref{Tab: Ciminion complexity} Bariant's attack is evaluated for sample we round numbers.

    Note that the GCD step implicitly assumes that $2^{r_C + r_E - 1} < q$, though if the degree of the \Ciminion function exceeds $q$, then an exhaustive search over the truncated output always performs better than the GCD computation.

    \subsubsection{Eigenvalue Method}
    For standard \Ciminion it is sufficient to find a $\Fq$-valued root for $x$, because we can then invert $p_C \circ p_E$ to obtain a key guess.
    Obviously, our \Ciminion DRL Gr\"obner basis (\Cref{Th: Ciminion Groebner basis}) satisfies the structural assumption from \Cref{Sec: solving structured systems}, so the complexity of constructing a univariate polynomial in $x$ is, see \Cref{Equ: characteristic polynomial complexity},
    \begin{equation}
        \mathcal{O} \left( 2^{\omega \cdot (r_C + r_E - 2) + 1} \right).
    \end{equation}
    However, for standard \Ciminion this complexity is always dominated by Bariant's attack, see \Cref{Sec: Bariant attack complexity}.

    For \Ciminiontwo on the other hand solving for $x$ is not sufficient, because we cannot invert $p_C \circ p_E$ anymore due to additional key additions.
    Therefore, for \Ciminiontwo we have to find a $\Fq$-valued solution for all quadratic variables to derive a key guess.
    Since the \Ciminiontwo DRL Gr\"obner basis satisfies the structure of \Cref{Equ: special shape}, the complexity of the Eigenvalue Method is given by \Cref{Equ: complexity solving structured systems}
    \begin{equation}\label{Equ: Ciminion Eigenvalue Method complexity}
        \mathcal{O} \left(
        \begin{rcases}
            \begin{dcases}
                2^{\omega + 1} \cdot \frac{2^{\omega \cdot (r_C + r_E - 1)} - N^\frac{r_C + r_E - 1}{2}}{2^{2 \cdot \omega} - N}, & 2 \mid (r_C + r_E - 1), \\
                2^{2 \cdot \omega + 1} \cdot \frac{2^{\omega \cdot (r_C + r_E - 2)} - N^\frac{r_C + r_E - 2}{2}}{2^{2 \cdot \omega} - N}, & 2 \nmid (r_C + r_E - 1)
            \end{dcases}
        \end{rcases}
        \right),
    \end{equation}
    where $N$ is a bound on the number of $\Fq$-valued solutions in every iteration.
    In addition, we absorbed the terms for the GCDs with the field equations into the implied constant, because these steps are always dominated by the construction of the respective characteristic polynomial.
    Since a bound on $N$ is in general not known we consider a worst-case scenario approach.
    We assume that the first $\Fq$-valued solution that an adversary recovers in every iteration leads to the true key, i.e.\ we assume that $N = 1$ in \Cref{Equ: Ciminion Eigenvalue Method complexity}.

    In \Cref{Tab: Ciminion complexity} we evaluated the discussed \Ciminion complexity estimations.
    All our estimations are from a designer's point of view, which assume an ideal adversary who has access to a hypothetical matrix multiplication algorithm which achieves $\omega = 2$.
    As already mentioned, for standard \Ciminion Bariant's attack is the most performative attack, and at least $r_C + r_E \geq 112$ is required to achieve $128$ bits of security.
    For \Ciminiontwo on the other hand this attack is not feasible, so we have to fall back to the Eigenvalue Method.
    Then, we require at least $r_C + r_E \geq 66$ to achieve $128$ bits of security.
    Finally, the estimation of the \Ciminion designers grossly underestimates the capabilities of Gr\"obner basis attacks.
    After all, we can construct an iterated \Ciminion DRL Gr\"obner basis in negligible time, whereas the designers estimated this complexity to be in $\mathcal{O} \left( 2^{2 \cdot \omega \cdot (r_C + r_E)} \right)$.

    \begin{table}[H]
        \centering
        \caption{Complexity estimations for various Gr\"obner basis attacks on \Ciminion and \Ciminiontwo with $q = 2^{127} + 45$ and $\omega = 2$.}
        \label{Tab: Ciminion complexity}
        \begin{tabular}{ c | M{25mm} | c || M{37mm} }
            \toprule
            & \multicolumn{3}{ c }{Complexity (bits)} \\
            \midrule
            $r_C + r_E$ & Bariant's Attack \cite{Bariant-Ciminion} & Eigenvalue Method & Fully Substituted Model \cite[\S 4.4]{EC:DGGK21} \\
            \midrule

            $32$  & $45.60$  & $61.09$  & $126$ \\
            $33$  & $46.67$  & $63.09$  & $130$ \\
            $65$  & $80.18$  & $127.09$ & $258$ \\
            $66$  & $81.22$  & $129.09$ & $262$ \\
            $111$ & $127.45$ & $219.09$ & $442$ \\
            $112$ & $128.47$ & $221.09$ & $446$ \\

            \bottomrule
        \end{tabular}
    \end{table}

    Recall from \cite[Table~1]{EC:DGGK21} that $r_C \in \left\{ s + 6, \ceil{\frac{2 \cdot (s + 6)}{3}} \right\}$, where $s$ is the required security level.
    For \Ciminiontwo this already implies that $\dim_{\Fq} \left( \mathcal{F}_\Ciminiontwo \right) \geq 2^\frac{s}{2}$, and the complexity of the Eigenvalue Method will scale approximately with $\gtrsim 2^{\omega \cdot \frac{s}{2}} \geq 2^s$.
    So, for now we conclude that \Ciminiontwo resists Gr\"obner basis attacks within the security level, even under ideal assumptions on the adversary.

    \section{Analysis of \Hydra}\label{Sec: Analysis of Hydra}
    The \Hydra heads are based on the Lai--Massey permutation.
    For the Gr\"obner basis analysis we recall that the Lai--Massey permutation
    \begin{equation}
        \mathcal{L}:
        \begin{pmatrix}
            x \\ y
        \end{pmatrix}
        \mapsto
        \begin{pmatrix}
            x + F (x - y) \\
            y + F (x - y)
        \end{pmatrix}
    \end{equation}
    is linearly equivalent to the Feistel permutation
    \begin{equation}
        \mathcal{F}:
        \begin{pmatrix}
            x \\ y
        \end{pmatrix}
        \mapsto
        \begin{pmatrix}
            x \\
            y + F (x)
        \end{pmatrix}
    \end{equation}
    via the matrices $\mathbf{A}_1 =
    \begin{pmatrix}
        1 & -1 \\
        0 & 1
    \end{pmatrix}
    $ and $\mathbf{A}_2 =
    \begin{pmatrix}
        1 & 1 \\
        0 & 1 \\
    \end{pmatrix}
    $, i.e.\ $\mathcal{L} = \mathbf{A}_2 \circ \mathcal{F} \circ \mathbf{A}_1$.
    Inspired by this equivalence we can reduce the \Hydra polynomial system to only one non-linear equation per round, i.e.\ the \Hydra polynomial system transforms to a Feistel-like polynomial system.

    In \cite[\S 6.2]{ToSC:Steiner24} the iterated polynomial systems representing generalized Feistel ciphers were analyzed.
    For these polynomial systems linear systems were derived, and in case the systems are of full rank the iterated polynomial systems are in generic coordinates, see \cite[Theorem~6.5]{ToSC:Steiner24}.
    For a system in generic coordinates we then have a proven upper bound on the solving degree, see \Cref{Th: solvdeg and CM-regularity}.

    Based on the aforementioned linear equivalence, we are going to replicate this strategy to construct a linear system for \Hydra.
    Further, if the \Hydra linear system has full rank, then we can transform the system into a DRL Gr\"obner basis via a linear change of coordinates.

    \subsection{Generic Coordinates for \Hydra}\label{Sec: linear system for Hydra}
    Recall that the \Hydra iterated polynomial system (\Cref{Def: Hydra polynomial system}) consists of $16 \cdot r_\mathcal{H} + 8$ polynomials in $16 \cdot r_\mathcal{H} + 4$ variables.
    Let
    \begin{equation}
        \mathbf{A}_n =
        \begin{pmatrix}
            1 & 0 & \ldots & 0 & -1 \\
            0 & 1 & \ldots & 0 & -1 \\
            \vdots & \vdots & \ddots & \vdots & \vdots \\
            0 & 0 & \ldots & 1 & -1 \\
            0 & 0 & \ldots & 0 & 1
        \end{pmatrix}
        \in \GL{n}{\Fq},
    \end{equation}
    and let $\mathbf{B} = \diag \left( \mathbf{A}_4, \mathbf{A}_4 \right) \in \GL{8}{\Fq}$.
    Then the first seven components of $\mathbf{A}_8 \mathbf{M}_\mathcal{J}^{-1} \mathbf{f}_j^{(i)}$, where $1 \leq i \leq r_\mathcal{H}$ and $j \in \{ 1, 2 \}$, are affine and the eighth component is quadratic.
    Moreover, the fourth and the eighth components of $\mathbf{B} \mathbf{M}_\mathcal{R}^{-1} \mathbf{f}_\mathcal{R}$ are quadratic while the other ones are affine.
    In particular, by extracting the $2 \cdot 7 \cdot r_\mathcal{H}$ affine components of the $\mathbf{A}_8 \mathbf{M}_\mathcal{J}^{-1} \mathbf{f}_j^{(i)}$'s we can eliminate the variables
    \begin{equation}
        y_1, \dots, y_4, z_1, \dots, z_3, x_{j, 1}^{(i)}, \dots, x_{j, 7}^{(i)},
    \end{equation}
    where $0 \leq i \leq r_\mathcal{H} - 1$ and $j \in \{ 1, 2 \}$, in the quadratic equations.
    (This can be explicitly seen in the proof of \Cref{Th: Hydra generic coordinates test}.)
    Thus, after the elimination we reduce to $2 \cdot r_\mathcal{H} + 8$ polynomials, of which $6$ are also affine, in $2 \cdot r_\mathcal{H} + 4$ variables.

    Next we derive a linear system to verify being in generic coordinates for \Hydra.
    \begin{thm}\label{Th: Hydra generic coordinates test}
        Let $\Fq$ be a finite field, let $\overline{\Fq}$ be its algebraic closure, let
        \[
            \mathcal{F}_\Hydra = \left\{ \mathbf{f}_1^{(i)}, \mathbf{f}_\mathcal{R}, \mathbf{f}_2^{(i)} \right\}_{1 \leq i \leq r_\mathcal{H}} \subset P = \overline{\Fq} \Big[ \mathbf{y}, \mathbf{z}, \mathbf{x}_1^{(1)}, \dots, \mathbf{x}_1^{(r_\mathcal{H} - 1)}, \mathbf{x}_2^{(0)}, \dots, \mathbf{x}_2^{(r_\mathcal{H} - 1)}, \mathbf{k} \Big]
        \]
        be a \Hydra polynomial system, and let
        \[
            \mathcal{G}_\Hydra = \left\{ \mathbf{A}_8 \mathbf{M}_\mathcal{J}^{-1} \mathbf{f}_1^{(i)}, \mathbf{B} \mathbf{M}_\mathcal{R}^{-1} \mathbf{f}_\mathcal{R}, \mathbf{A}_8 \mathbf{M}_\mathcal{J}^{-1} \mathbf{f}_2^{(i)} \right\}_{1 \leq i \leq r_\mathcal{H}}.
        \]
        Then every homogeneous ideal $I \subset P [x_0]$ such that $x_0 \notin \sqrt{I}$ and $\mathcal{G}_\Hydra^\homog \subset I$ is in generic coordinates if the following linear system has rank $16 \cdot r_\mathcal{H} + 4$
            \begin{align}
            y_1 - z_4 + \bigg( \mathbf{M}_\mathcal{J}^{-1} \Big( \mathbf{k}' - \mathbf{x}_1^{(1)} \Big) \bigg)_1 + \bigg( \mathbf{M}_\mathcal{J}^{-1} \Big( \mathbf{x}_1^{(1)} - \mathbf{k}' \Big) \bigg)_8 &= 0, \nonumber \\
            &\vdots \nonumber \\
            z_3 - z_4 + \bigg( \mathbf{M}_\mathcal{J}^{-1} \Big( \mathbf{k}'- \mathbf{x}_1^{(1)} \Big) \bigg)_7 + \bigg( \mathbf{M}_\mathcal{J}^{-1} \Big( \mathbf{x}_1^{(1)} - \mathbf{k}' \Big) \bigg)_8 &= 0, \nonumber \\
            \sum_{l = 1}^{4} y_l - z_l &= 0, \nonumber \\
            x_{2, k}^{(0)} - x_{2, 8}^{(0)} + \bigg( \mathbf{M}_\mathcal{J}^{-1} \Big( \mathbf{k}' - \mathbf{x}_2^{(1)} \Big) \bigg)_k + \bigg( \mathbf{M}_\mathcal{J}^{-1} \Big( \mathbf{x}_2^{(1)} - \mathbf{k}' \Big) \bigg)_8 &= 0, \nonumber \\
            \sum_{l = 1}^{8} (-1)^{\floor{\frac{l - 1}{4}}} \cdot x_{2, l}^{(0)} &= 0, \nonumber \\
            x_{j, k}^{(i - 1)} - x_{j, 8}^{(i - 1)} + \bigg( \mathbf{M}_\mathcal{J}^{-1} \Big( \mathbf{k}' - \mathbf{x}_j^{(i)} \Big) \bigg)_k + \bigg( \mathbf{M}_\mathcal{J}^{-1} \Big( \mathbf{x}_j^{(i)} - \mathbf{k}' \Big) \bigg)_8 &= 0, \nonumber \\
            \sum_{l = 1}^{8} (-1)^{\floor{\frac{l - 1}{4}}} \cdot x_{j, l}^{(i - 1)} &= 0, \nonumber \\
            x_{1, k}^{(r_\mathcal{H} - 1)} - x_{1, 8}^{(r_\mathcal{H} - 1)} + \bigg( \mathbf{M}_\mathcal{J}^{-1} \Big( \mathbf{k}' + (\mathbf{y}, \mathbf{z})^\intercal \Big) \bigg)_k - \bigg( \mathbf{M}_\mathcal{J}^{-1} \Big(\mathbf{k}' + (\mathbf{y}, \mathbf{z})^\intercal \Big) \bigg)_8 &= 0, \nonumber \\
            \sum_{l = 1}^{8} (-1)^{\floor{\frac{l - 1}{4}}} \cdot x_{1, l}^{(r_\mathcal{H} - 1)} &= 0, \nonumber \\
            x_{2, k}^{(r_\mathcal{H} - 1)} - x_{2, 8}^{(r_\mathcal{H} - 1)} + \bigg( \mathbf{M}_\mathcal{J}^{-1} \Big( \mathbf{k}' + \mathbf{x}_2^{(0)} \Big) \bigg)_k - \bigg( \mathbf{M}_\mathcal{J}^{-1} \Big( \mathbf{k}' + \mathbf{x}_2^{(0)} \Big) \bigg)_8 &= 0, \nonumber \\
            \sum_{l = 1}^{8} (-1)^{\floor{\frac{l - 1}{4}}} \cdot x_{2, l}^{(r_\mathcal{H} - 1)} &= 0, \nonumber \\
            y_1 - y_4 - \Big( \mathbf{M}_\mathcal{R}^{-1} \mathbf{x}_2^{(0)} \Big)_1 + \Big( \mathbf{M}_\mathcal{R}^{-1} \mathbf{x}_2^{(0)} \Big)_4 &= 0, \nonumber \\
            &\vdots \nonumber \\
            y_3 - y_4 - \Big( \mathbf{M}_\mathcal{R}^{-1} \mathbf{x}_2^{(0)} \Big)_3 + \Big( \mathbf{M}_\mathcal{R}^{-1} \mathbf{x}_2^{(0)} \Big)_4 &= 0, \nonumber \\
            z_1 - z_4 - \Big( \mathbf{M}_\mathcal{R}^{-1} \mathbf{x}_2^{(0)} \Big)_5 + \Big( \mathbf{M}_\mathcal{R}^{-1} \mathbf{x}_2^{(0)} \Big)_8 &= 0, \nonumber \\
            &\vdots \nonumber \\
            z_3 - z_4 - \Big( \mathbf{M}_\mathcal{R}^{-1} \mathbf{x}_2^{(0)} \Big)_7 + \Big( \mathbf{M}_\mathcal{R}^{-1} \mathbf{x}_2^{(0)} \Big)_8 &= 0, \nonumber
        \end{align}
        where $2 \leq i \leq r_\mathcal{H} - 1$, $j \in \{ 1, 2 \}$ and $1 \leq k \leq 7$.
    \end{thm}
    \begin{proof}
        By \cite[Theorem~3.2]{ToSC:Steiner24} $\left( \mathcal{G}_\Hydra^\homog \right)$ is in generic coordinates if
        \begin{equation}\label{Equ: Hydra highest degree components radical ideal}
            \sqrt{\mathcal{G}_\Hydra^\topcomp} = \left( \mathbf{y}, \mathbf{z}, \mathbf{x}_1^{(1)}, \dots, \mathbf{x}_1^{(r_\mathcal{H} - 1)}, \mathbf{x}_2^{(0)}, \dots, \mathbf{x}_2^{(r_\mathcal{H} - 1)}, \mathbf{k} \right).
        \end{equation}
        To shorten some writing let
        \begin{align}
            \mathbf{g}_j^{(i)} &= \left( \mathbf{A}_8 \mathbf{M}_\mathcal{J}^{-1} \mathbf{f}_j^{(i)} \right)^\homog \mod (x_0), \nonumber \\
            \mathbf{g}_\mathcal{R} &= \left( \mathbf{B} \mathbf{M}_\mathcal{R}^{-1} \mathbf{f}_\mathcal{R} \right)^\homog \mod (x_0). \nonumber
        \end{align}
        Then for $i = 1$ we have that
        \begin{align}
            \mathbf{g}_1^{(1)} &=
            \begin{pmatrix}
                \left\{ y_k - z_4 + \bigg( \mathbf{M}_\mathcal{J}^{-1} \Big( \mathbf{k}' - \mathbf{x}_1^{(1)} \Big) \bigg)_k + \bigg( \mathbf{M}_\mathcal{J}^{-1} \Big( \mathbf{x}_1^{(1)} - \mathbf{k}' \Big) \bigg)_8 \right\}_{1 \leq k \leq 7} \\[10pt]
                \left( \sum_{l = 1}^{4} y_l - z_l \right)^2
            \end{pmatrix}
            , \nonumber \\
            \mathbf{g}_2^{(1)} &=
            \begin{pmatrix}
                \left\{ x_{2, k}^{(0)} - x_{2, 8}^{(0)} + \bigg( \mathbf{M}_\mathcal{J}^{-1} \Big( \mathbf{k}' - \mathbf{x}_2^{(1)} \Big) \bigg)_k + \bigg( \mathbf{M}_\mathcal{J}^{-1} \Big( \mathbf{x}_2^{(1)} - \mathbf{k}' \Big) \bigg)_8 \right\}_{1 \leq k \leq 7} \\[10pt]
                \left( \sum_{l = 1}^{8} (-1)^{\floor{\frac{l - 1}{4}}} \cdot x_{2, l}^{(0)} \right)^2
            \end{pmatrix}
            . \nonumber
        \end{align}
        For $2 \leq i \leq r_\mathcal{H} - 1$ and $j \in \{ 1, 2 \}$ we have that
        \[
            \mathbf{g}_j^{(i)} =
            \begin{pmatrix}
                \left\{ x_{j, k}^{(i - 1)} - x_{j, 8}^{(i - 1)} + \bigg( \mathbf{M}_\mathcal{J}^{-1} \Big( \mathbf{k}' - \mathbf{x}_j^{(i)} \Big) \bigg)_k + \bigg( \mathbf{M}_\mathcal{J}^{-1} \Big( \mathbf{x}_j^{(i)} - \mathbf{k}' \Big) \bigg)_8 \right\}_{1 \leq k \leq 7} \\[10pt]
                \left( \sum_{l = 1}^{8} (-1)^{\floor{\frac{l - 1}{4}}} \cdot x_{j, l}^{(i - 1)} \right)^2
            \end{pmatrix}
            .
        \]
        For $i = r_\mathcal{H}$ we have that
        \begin{align}
            &\mathbf{g}_1^{(r_\mathcal{H})} = \nonumber \\
            &
            \begin{pmatrix}
                \left\{ x_{1, k}^{(r_\mathcal{H} - 1)} - x_{1, 8}^{(r_\mathcal{H} - 1)} + \bigg( \mathbf{M}_\mathcal{J}^{-1} \Big( \mathbf{k}' + (\mathbf{y}, \mathbf{z})^\intercal \Big) \bigg)_k - \bigg( \mathbf{M}_\mathcal{J}^{-1} \Big(\mathbf{k}' + (\mathbf{y}, \mathbf{z})^\intercal \Big) \bigg)_8 \right\}_{1 \leq k \leq 7} \\[10pt]
                \left( \sum_{l = 1}^{8} (-1)^{\floor{\frac{l - 1}{4}}} \cdot x_{1, l}^{(r_\mathcal{H} - 1)} \right)^2
            \end{pmatrix}
            , \nonumber \\
            &\mathbf{g}_2^{(r_\mathcal{H})} = \nonumber \\
            &
            \begin{pmatrix}
                \left\{ x_{2, k}^{(r_\mathcal{H} - 1)} - x_{2, 8}^{(r_\mathcal{H} - 1)} + \bigg( \mathbf{M}_\mathcal{J}^{-1} \Big( \mathbf{k}' + \mathbf{x}_2^{(0)} \Big) \bigg)_k - \bigg( \mathbf{M}_\mathcal{J}^{-1} \Big( \mathbf{k}' + \mathbf{x}_2^{(0)} \Big) \bigg)_8 \right\}_{1 \leq k \leq 7} \\[10pt]
                \left( \sum_{l = 1}^{8} (-1)^{\floor{\frac{l - 1}{4}}} \cdot x_{2, l}^{(r_\mathcal{H} - 1)} \right)^2
            \end{pmatrix}
            . \nonumber
        \end{align}
        Finally, for the rolling function we have that
        \[
            \mathbf{g}_\mathcal{R} =
            \begin{pmatrix}
                \left\{ y_k - y_4 - \Big( \mathbf{M}_\mathcal{R}^{-1} \mathbf{x}_2^{(0)} \Big)_k + \Big( \mathbf{M}_\mathcal{R}^{-1} \mathbf{x}_2^{(0)} \Big)_4 \right\}_{1 \leq k \leq 3} \\[10pt]
                g (\mathbf{y}, \mathbf{z}) \\
                \left\{ z_k - z_4 - \Big( \mathbf{M}_\mathcal{R}^{-1} \mathbf{x}_2^{(0)} \Big)_{4 + k} + \Big( \mathbf{M}_\mathcal{R}^{-1} \mathbf{x}_2^{(0)} \Big)_8 \right\}_{1 \leq k \leq 3} \\[10pt]
                g (\mathbf{z}, \mathbf{y})
            \end{pmatrix}
            .
        \]
        Except for the rolling function, all quadratic polynomials are squares of linear polynomials.
        Collecting the linear polynomials from $\mathcal{G}_\Hydra^\topcomp$ together with the linear polynomials from the squares, we yield the linear system from the assertion.
        It is clear that all these linear polynomials are elements of $\sqrt{\mathcal{G}_\Hydra^\topcomp}$.
        Since the highest degree components are homogeneous, we always have that $\sqrt{\mathcal{G}_\Hydra^\topcomp} \subset \left( \mathbf{y}, \mathbf{z}, \mathbf{x}_1^{(i)}, \dots, \mathbf{x}_1^{(r_\mathcal{H} - 1)}, \mathbf{x}_2^{(0)}, \dots, \mathbf{x}_2^{(r_\mathcal{H} - 1)}, \mathbf{k} \right)$.
        Thus, if the linear system from the assertion has full rank, then we have established the equality from \Cref{Equ: Hydra highest degree components radical ideal}.
    \end{proof}

    Naturally, this theorem can be extended to an arbitrary number of \Hydra samples.
    Now it is straight-forward to compute the Macaulay bound (\Cref{Th: solvdeg and CM-regularity}) for \Hydra.
    \begin{cor}\label[cor]{Cor: Hydra Macaulay bound}
        Let $\Fq$ be a finite field, let $\overline{\Fq}$ be its algebraic closure, let
        \[
            \mathcal{F}_\Hydra = \left\{ \mathbf{f}_1^{(i)}, \mathbf{f}_\mathcal{R}, \mathbf{f}_2^{(i)} \right\}_{1 \leq i \leq r_\mathcal{H}} \subset P = \overline{\Fq} \Big[ \mathbf{y}, \mathbf{z}, \mathbf{x}_1^{(1)}, \dots, \mathbf{x}_1^{(r_\mathcal{H} - 1)}, \mathbf{x}_2^{(0)}, \dots, \mathbf{x}_2^{(r_\mathcal{H} - 1)}, \mathbf{k} \Big]
        \]
        be a \Hydra polynomial system.
        Let
        \[
            \mathcal{G}_{\Hydra} = \left\{ \mathbf{A}_{8} \mathbf{M}_{\mathcal{J}}^{-1} \mathbf{f}_{1}^{(i)}, \mathbf{B} \mathbf{M}_\mathcal{R}^{-1} \mathbf{f}_\mathcal{R}, \mathbf{A}_{8} \mathbf{M}_{\mathcal{J}}^{-1} \mathbf{f}_{2}^{(i)} \right\}_{1 \leq i \leq r_{\mathcal{H}}},
        \]
        and assume that $\left( \mathcal{G}_{\Hydra}^{\homog} \right)$ is in generic coordinates.
        If the affine polynomials in $\mathcal{G}_{\Hydra}$ have rank $2 \cdot 7 \cdot r_{\mathcal{H}} + 6$, then
        \[
            \solvdeg_{DRL} \left( \mathcal{G}_{\Hydra} \right) \leq 2 \cdot r_{\mathcal{H}}.
        \]
    \end{cor}
    \begin{proof}
        If the affine polynomials have rank $2 \cdot 7 \cdot r_{\mathcal{H}} + 6$, then we can reduce to $2 \cdot r_\mathcal{H} + 2$ quadratic polynomials in $2 \cdot r_\mathcal{H} - 2$ variables.
        This variable elimination does not affect being in generic coordinates, so the Macaulay bound (\Cref{Th: solvdeg and CM-regularity}) evaluates to
        \[
            \solvdeg_{DRL} \left( \mathcal{G}_{\Hydra} \right) \leq 2 \cdot (2 \cdot r_{\mathcal{H}} - 1) - \left( 2 \cdot r_{\mathcal{H}} - 1 \right) + 1 = 2 \cdot r_{\mathcal{H}}. \qedhere
        \]
    \end{proof}

    \subsection{Extracting a Zero-Dimensional DRL Gr\"obner Basis for \Hydra}\label{Sec: Extracting a Groebner Basis}
    Due to the structure of the \Hydra heads we can actually produce a DRL Gr\"obner basis via a simple linear change of coordinates.
    Again, let $\mathcal{G}_\Hydra$ be as in \Cref{Th: Hydra generic coordinates test}, we assume that
    \begin{enumerate}[label=(\roman*)]
        \item\label{Item: maximal rank} the  linear system from \Cref{Th: Hydra generic coordinates test} has full rank, and

        \item the linear polynomials in $\mathcal{G}_{\Hydra}^\topcomp$ have rank $2 \cdot 7 \cdot r_{\mathcal{H}} + 6$.
    \end{enumerate}
    I.e., $\left( \mathcal{G}_\Hydra^\homog \right)$ is in generic coordinates, and we can use the linear polynomials in $\mathcal{G}_\Hydra$ to reduce the system to $2 \cdot r_\mathcal{H} + 2$ quadratic polynomials in $2 \cdot r_\mathcal{H} - 2$ variables.

    By construction $2 \cdot r_\mathcal{H}$ of these quadratic polynomials come from the \Hydra heads and their quadratic terms are squares of linear equations.
    In particular, before the variable elimination the quadratic terms are, see \Cref{Def: Hydra polynomial system} and the proof of \Cref{Th: Hydra generic coordinates test},
    \begin{alignat}{2}
    	g_{1, 8}^{(1)} &= \left( \sum_{l = 1}^{4} y_l - z_l \right)^2 &&= \mathcal{L}_{1, 1}^2, \label{Equ: linear equations 1} \\
    	g_{2, 8}^{(1)} &= \left( \sum_{l = 1}^{8} (-1)^{\floor{\frac{l - 1}{4}}} \cdot x_{2, l}^{(0)} \right)^2 &&= \mathcal{L}_{2, 1}^2, \label{Equ: linear equations 2} \\
    	g_{j, 8}^{(i)} &= \left( \sum_{l = 1}^{8} (-1)^{\floor{\frac{l - 1}{4}}} \cdot x_{j, l}^{(i - 1)} \right)^2 &&= \mathcal{L}_{j, i}^2, \qquad 2 \leq i \leq r_\mathcal{H},\ j \in \{ 1, 2 \}, \label{Equ: linear equations 3}
    \end{alignat}
    Recall that these linear equations $\mathcal{L} = \{ \mathcal{L}_{1, i}, \mathcal{L}_{2, i} \}_{1 \leq i \leq 2 \cdot r_\mathcal{H}}$ also appeared in the linear system from \Cref{Th: Hydra generic coordinates test}.
    By Assumption \ref{Item: maximal rank}, $\mathcal{L}$ together with the linear equations from $\mathcal{G}_\Hydra^\topcomp$ has rank $16 \cdot r_\mathcal{H} + 4$.
    After eliminating $2 \cdot 7 \cdot r_{\mathcal{H}} + 6$ variables with the affine equations from $\mathcal{G}_\Hydra$, the polynomials from \Cref{Equ: linear equations 1,Equ: linear equations 2,Equ: linear equations 3} modify to
    \begin{equation}\label{Equ: linear equations 4}
        g_{j, 8}^{(i)} = \left( \tilde{\mathcal{L}}_{j, i} + a_{j, i} \right)^2, \qquad 1 \leq i \leq r_\mathcal{H},\ j \in \{ 1, 2 \},
    \end{equation}
    where $\tilde{\mathcal{L}}_{j, i}$ is a linear polynomial in the remaining $2 \cdot r_\mathcal{H} - 2$ variables and $a_{j, i} \in \Fq$.
    Moreover, the rank of a matrix does not depend on the order of Gaussian elimination, so by Assumption \ref{Item: maximal rank} the new linear polynomials $\tilde{\mathcal{L}} = \left\{ \tilde{\mathcal{L}}_{1, i}, \tilde{\mathcal{L}}_{2, i} \right\}_{1 \leq i \leq 2 \cdot r_\mathcal{H}}$ still must contain a maximal linearly independent subset of rank $2 \cdot r_\mathcal{H} - 2$.

    Let $\mathbf{x}_\text{nl}$ denote the remaining $2 \cdot r_\mathcal{H} - 2$ non-linear variables.
    We can now collect this maximal linearly independent subset of $\tilde{\mathcal{L}}$ in a matrix $\mathbf{M} \in \Fq^{(2 \cdot r_\mathcal{H} - 2) \times (2 \cdot r_\mathcal{H} - 2)}$.
    Columns of $\mathbf{M}$ are indexed by the elements of $\mathbf{x}_\text{nl}$, and in the rows of $\mathbf{M}$ we write the coefficients of the maximal linearly independent equations.
    Next we pick new variables $\mathbf{\hat{x}} = \left( \hat{x}_1, \dots, \hat{x}_{2 \cdot r_\mathcal{H} - 2} \right)^\intercal$, and finally we perform a linear change of variables $\mathbf{\hat{x}} = \mathbf{M} \mathbf{x}_\text{nl}$.
    By construction $2 \cdot r_\mathcal{H}$ polynomials in $\mathcal{G}_\Hydra$ are of the form
    \begin{equation}
        \tilde{\mathcal{L}}_{j, i} (\mathbf{x}_\text{nl})^2 + \mathcal{A}_{j, i} (\mathbf{x}_\text{nl}) = 0,
    \end{equation}
    where $\mathcal{A}_{j, i} \in \Fq [\mathbf{x}_\text{nl}]$ is affine, $1 \leq i \leq 2 \cdot r_\mathcal{H}$ and $j \in \{1, 2\}$.
    Therefore, after the change of coordinates we produce $2 \cdot r_\mathcal{H} - 2$ polynomials of the form
    \begin{equation}
        g_i = \hat{x}_i^2 + \sum_{j = 1}^{2 \cdot r_\mathcal{H} - 2} \alpha_{i, j} \cdot \hat{x}_j + \alpha_{i, 0},
    \end{equation}
    where $\alpha_{i, j}, \alpha_{i, 0} \in \Fq$, together with four additional quadratic polynomials.
    In particular, $\hat{\mathcal{G}} = \left\{ g_i \right\}_{1 \leq i \leq 2 \cdot r_\mathcal{H} - 2} \subset \Fq \big[ \mathbf{\hat{x}} \big]$ is a DRL Gr\"obner basis with $\dim_{\Fq} \left( \hat{\mathcal{G}} \right) = 2^{2 \cdot r_{\mathcal{H}} - 2}$.

    We summarize this construction in a more formal manner in the next corollary.
    \begin{cor}\label[cor]{Cor: Hydra DRL Groebner basis}
        Let $\Fq$ be a finite field, let
        \[
            \mathcal{F}_\Hydra = \left\{ \mathbf{f}_1^{(i)}, \mathbf{f}_\mathcal{R}, \mathbf{f}_2^{(i)} \right\}_{1 \leq i \leq r_\mathcal{H}} \subset P = \Fq \Big[ \mathbf{y}, \mathbf{z}, \mathbf{x}_1^{(1)}, \dots, \mathbf{x}_1^{(r_\mathcal{H} - 1)}, \mathbf{x}_2^{(0)}, \dots, \mathbf{x}_2^{(r_\mathcal{H} - 1)}, \mathbf{k} \Big]
        \]
        be a \Hydra polynomial system, and let
        \begin{align*}
            \mathcal{G}_{\Hydra} &= \left\{ \mathbf{A}_{8} \mathbf{M}_{\mathcal{J}}^{-1} \mathbf{f}_{1}^{(i)}, \mathbf{B} \mathbf{M}_\mathcal{R}^{-1} \mathbf{f}_\mathcal{R}, \mathbf{A}_{8} \mathbf{M}_{\mathcal{J}}^{-1} \mathbf{f}_{2}^{(i)} \right\}_{1 \leq i \leq r_{\mathcal{H}}}, \\
            \mathcal{G}_{lin} &= \left\{ g \in \mathcal{G}_{\Hydra} \mid \degree{g} = 1 \right\}, \\
            \mathcal{G}_{squ} &= \left\{ g \in \mathcal{G}_{\Hydra} \mid \degree{g} = 2 \right\}.
        \end{align*}
        Assume that
        \begin{enumerate}[label=(\roman*)]
            \item the linear system from \Cref{Th: Hydra generic coordinates test} has full rank, and

            \item $\rank \left( \mathcal{G}_{lin}^{\topcomp} \right) = 2 \cdot 7 \cdot r_{\mathcal{H}} + 6$.
        \end{enumerate}
        Then there exists a linear change of coordinates $\mathcal{G}_{sub} \subset P \big[ \hat{x}_{i} \; \big\vert \; 1 \leq i \leq 2 \cdot r_{\mathcal{H}} - 2 \big]$ such that:
        \begin{enumerate}
            \item $\hat{\mathcal{G}}_{squ} = \mathcal{G}_{squ} \mod \left( \mathcal{G}_{lin}, \mathcal{G}_{sub} \right) \subset \Fq \big[ \hat{x}_{i} \; \big\vert \; 1 \leq i \leq 2 \cdot r_{\mathcal{H}} - 2 \big]$, where division is performed with respect to the DRL term order $\mathbf{y} > \mathbf{z} > \mathbf{x}_1^{(1)} > \ldots > \mathbf{x}_1^{(r_\mathcal{H} - 1)} > \mathbf{x}_2^{(0)} > \ldots > \mathbf{x}_2^{(r_\mathcal{H} - 1)} > \mathbf{k} > \mathbf{\hat{x}}$.

            \item For every $1 \leq i \leq 2 \cdot r_{\mathcal{H}} - 2$ there exists $g_{i} \in \hat{\mathcal{G}}_{squ}$ such that
            \begin{equation*}
                g_{i} = \hat{x}_{i}^{2} + \mathcal{A}_{i},
            \end{equation*}
            where $\mathcal{A}_{i} \in \Fq \big[ \hat{x}_{i} \; \big\vert \; 1 \leq i \leq 2 \cdot r_{\mathcal{H}} - 2 \big]$ is affine.

            \item $\dim_{\Fq} \left( \mathcal{G}_{lin} \cup \mathcal{G}_{sub} \cup \hat{\mathcal{G}}_{squ} \right) \leq 2^{2 \cdot r_{\mathcal{H}} - 2}$.
        \end{enumerate}
    \end{cor}

    Besides being tedious to do by hand, success of the change of coordinates depends on the matrices used in the \Hydra heads.
    Therefore, we did not attempt a full proof by hand, instead for specific instances we outsourced it to a computer algebra program where the construction can be performed in a matter of seconds.

    \subsection{Implementation}\label{Sec: Hydra implementation}
    We have implemented \Hydra and its iterated polynomial system in the computer algebra systems \SageMath \cite{SageMath} and \OSCAR \cite{OSCAR}.\footnote{\repository}
    In both systems we implemented the linear system from \Cref{Th: Hydra generic coordinates test} and the change of coordinates from \Cref{Sec: Extracting a Groebner Basis}.
    With the \SageMath implementation we can compute the DRL Gr\"obner basis for full rounds $r_\mathcal{H} = 39$ over $128$ bit prime fields in negligible time.
    The \OSCAR computer algebra system offers an efficient F4 implementation that can handle primes of size $q < 2^{31}$, therefore we recommend to use \OSCAR for experimental computations over small prime fields.

    For completeness, we mention that we implemented the rolling function without a constant addition, i.e.\ we set $\mathbf{c}_\mathcal{R}^{(i)} = \mathbf{0}$ for all $i \geq 1$.
    However, we emphasize that the constant addition does not affect either the linear system from \Cref{Th: Hydra generic coordinates test} or the change of coordinates from \Cref{Sec: Extracting a Groebner Basis}.

    \begin{ex}[A Concrete Instance]\label[ex]{Ex: Hydra instance}
        In \cite[Appendix~C]{EPRINT:GOSW22} the \Hydra designers propose concrete matrices over the prime field $q = 2^{127} + 45$
        \begin{align}
            \mathbf{M}_\mathcal{E} &= \circulant \left( 3, 2, 1, 1 \right), \qquad \qquad
            \mathbf{M}_\mathcal{I} =
            \begin{pmatrix}
                1 & 1 & 1 & 1 \\
                1 & 4 & 1 & 1 \\
                3 & 1 & 3 & 1 \\
                4 & 1 & 1 & 2
            \end{pmatrix}
            , \\
            \mathbf{M}_\mathcal{J} &=
            \begin{pmatrix}
                3 & 1 & 1 & 1 & 1 & 1 & 1 & 1 \\
                7 & 3 & 1 & 1 & 1 & 1 & 1 & 1 \\
                4 & 1 & 4 & 1 & 1 & 1 & 1 & 1 \\
                3 & 1 & 1 & 8 & 1 & 1 & 1 & 1 \\
                7 & 1 & 1 & 1 & 7 & 1 & 1 & 1 \\
                8 & 1 & 1 & 1 & 1 & 5 & 1 & 1 \\
                5 & 1 & 1 & 1 & 1 & 1 & 2 & 1 \\
                4 & 1 & 1 & 1 & 1 & 1 & 1 & 6
            \end{pmatrix}
            .
        \end{align}
        For security level $\kappa = 128$ bits the designers suggest to use $r_\mathcal{H} = 39$, see \cite[\S 5.5]{EC:GOSW23}.

        For this instance using our \SageMath implementation we verified that the \Hydra polynomial system is in generic coordinates for $m \in \{ 2, 3, 4, 5, 10 \}$ many samples.
        For $m = 2$ we can compute a \Hydra DRL Gr\"obner basis with the change of coordinates from \Cref{Sec: Extracting a Groebner Basis}.

        Recall the quadratic terms from the \Hydra heads, see \Cref{Equ: linear equations 1,Equ: linear equations 2,Equ: linear equations 3}.
        For $r_\mathcal{H} \geq 2$ and after the elimination of $2 \cdot 7 \cdot r_\mathcal{H} + 6$ variables, see \Cref{Equ: linear equations 4}, we observed that
        \begin{equation}
            \rank \left( \tilde{\mathcal{L}}_{1, 8}^{(3)}, \dots, \tilde{\mathcal{L}}_{1, 8}^{(r_\mathcal{H})}, \tilde{\mathcal{L}}_{2, 8}^{(1)}, \dots, \tilde{\mathcal{L}}_{2, 8}^{(r_\mathcal{H})} \right) = 2 \cdot r_\mathcal{H} - 2.
        \end{equation}
        So, before the elimination of $2 \cdot 7 \cdot r_\mathcal{H} + 6$ variables our change of coordinates corresponds to, see \Cref{Equ: linear equations 1,Equ: linear equations 2,Equ: linear equations 3},
        \begin{alignat}{2}
            \hat{x}_{i} &= \sum_{l = 1}^{8} (-1)^{\floor{\frac{l - 1}{4}}} \cdot x_{1, l}^{(i - 1)}, \qquad &&3 \leq i \leq r_\mathcal{H}, \\
            \hat{x}_{i + r_\mathcal{H} - 2} &= \sum_{l = 1}^{8} (-1)^{\floor{\frac{l - 1}{4}}} \cdot x_{2, l}^{(j - 1)}, \qquad &&1 \leq i \leq r_\mathcal{H}.
        \end{alignat}

        We also note that we did not find a round number $r_\mathcal{H} \geq 2$ for which the Gr\"obner basis extraction fails.
    \end{ex}

    \subsection{Cryptanalysis}\label{Sec: Hydra cryptanalysis}
    Next we discuss the cryptanalytic impact of our theoretical analysis.
    Throughout, this section we will always assume that we can perform the linear change of coordinates from \Cref{Sec: Extracting a Groebner Basis}.
    In particular, we stress that for the \Hydra instance from \Cref{Ex: Hydra instance} we are always able to construct a DRL Gr\"obner basis.
    Therefore, we have a quadratic \Hydra DRL Gr\"obner basis in $2 \cdot r_\mathcal{H} - 2$ equations and variables together with four additional quadratic polynomials at hand.
    This opens the door for two different attack strategies:
    \begin{enumerate}[label=(\Alph*)]
        \item Recomputation of the DRL Gr\"obner basis for the overdetermined system of $2 \cdot r_\mathcal{H} + 2$ quadratic polynomials.

        \item Polynomial system solving techniques for the DRL Gr\"obner basis of $2 \cdot r_\mathcal{H} - 2$ quadratic polynomials.
    \end{enumerate}
    For comparison, we also review the complexity estimation of the designers.

    \subsubsection{Estimation of the Designers}\label{Sec: Hydra designer estimate}
    The \Hydra designers considered the iterated polynomial model, see \cite[\S7.4]{EC:GOSW23} and \cite[Appendix~I.2]{EPRINT:GOSW22}, and reduced it to $2 \cdot r_\mathcal{H} + 2$ quadratic polynomials in $2 \cdot r_\mathcal{H} - 2$ variables.
    The designers assumed that this downsized system is semi-regular \cite{Bardet-Complexity}.
    Then, the degree of regularity $d_{\reg}$ is given by the index of the first negative coefficient in the power series
    \begin{equation}
        H (t)
        = \frac{(1 - t^2)^{2 \cdot r_\mathcal{H} + 2}}{(1 - t)^{2 \cdot r_\mathcal{H} - 2}},
    \end{equation}
    and the complexity of computing the DRL Gr\"obner basis is then bounded by \cite[Theorem~7]{Bardet-Complexity}
    \begin{equation}\label{Equ: Hydra semi-regular complexity}
    	\mathcal{O} \left( \binom{2 \cdot r_\mathcal{H} - 2 + d_{\reg}}{d_{\reg}}^\omega \right).
    \end{equation}
    For $r_\mathcal{H}^\ast = 29$ the designers estimate $d_{\reg} = 23$, with $\omega = 2$ this yields a complexity estimate of $\approx 2^{130.8}$ bits.
    Moreover, it is claimed that $r_\mathcal{H} = r_\mathcal{H}^\ast + 2 = 31$ rounds are sufficient to provide $128$ bits of security against Gr\"obner basis attacks.

    \subsubsection{Recomputing the DRL Gr\"obner Basis}
    After constructing the \Hydra DRL Gr\"obner basis we have one big advantage compared to the designer estimation: We yield a $\Fq$-Boolean polynomial system, see \Cref{Sec: Boolean Macaulay matrices}.
    Thus, in the Macaulay matrix we only have to consider square-free monomials.
    Effectively, this saves time and space.

    \paragraph{Proven Solving Degree Bound.}
    Recall from \Cref{Cor: Hydra Macaulay bound} that $\solvdeg_{DRL} \left( \mathcal{G}_\Hydra \right) \leq 2 \cdot r_\mathcal{H}$.
    So, in the worst case we have to construct the full $\Fq$-Boolean Macaulay matrix, see \Cref{Th: Boolean Macaulay matrix}.
    The complexity for constructing the $\Fq$-Boolean Macaulay matrix was given in \Cref{Cor: Boolean Macaulay matrix construction complexity}, for the \Hydra parameters $\abs{F} = 4$, $n, d = 2 \cdot r_\mathcal{H} - 2$ and $D = 2$ we yield the complexity
    \begin{equation}\label{Equ: Hydra Boolean Macaula matrix construction complexity}
        \mathcal{O} \left( 4 \cdot r_\mathcal{H} \cdot (2 \cdot r_\mathcal{H} - 1) \cdot \sum_{i = 0}^{2 \cdot r_\mathcal{H} - 2} \binom{2 \cdot r_\mathcal{H} + i}{i + 2} \cdot \binom{2 \cdot r_\mathcal{H} - 2}{i} \right).
    \end{equation}
    Further, by \Cref{Th: Boolean Macaulay matrix} Gaussian elimination on the $\Fq$-Boolean Macaulay matrix can be performed in $\mathcal{O} \left( 4 \cdot 2^{\omega \cdot (2 \cdot r_\mathcal{H} - 2)} \right)$ field operations.

    In \Cref{Tab: Hydra recomputing DRL Groebner Basis} we evaluated these complexities for sample round numbers.

    \paragraph{Semi-Regularity Assumption.}
    If we assume that the \Hydra polynomial system is semi-regular after the change of variables (\Cref{Sec: Extracting a Groebner Basis}), then \Cref{Equ: Hydra Boolean Macaula matrix construction complexity} improves to
    \begin{equation}
        \mathcal{O} \left( 4 \cdot r_\mathcal{H} \cdot (2 \cdot r_\mathcal{H} - 1) \cdot \sum_{i = 0}^{d_{\reg} - 2} \binom{2 \cdot r_\mathcal{H} - 2}{i} \cdot \binom{2 \cdot r_\mathcal{H} + i}{i + 2} \right),
    \end{equation}
    and the complexity for Gaussian elimination improves to
    \begin{equation}
        \mathcal{O} \left( 4 \cdot \left( \sum_{i = 0}^{d_{\reg} - 2} \binom{2 \cdot r_\mathcal{H} - 2}{i} \right) \cdot \left( \sum_{i = 0}^{d_{\reg}} \binom{2 \cdot r_\mathcal{H} - 2}{i} \right)^{\omega - 1} \right).
    \end{equation}

    In \Cref{Tab: Hydra recomputing DRL Groebner Basis} we evaluated the complexities together with the one from the designers (\Cref{Sec: Hydra designer estimate}).
    Construction of the $\Fq$-Boolean Macaulay matrix dominates this approach.
    As expected, under the semi-regularity assumption the Boolean approach is slightly more performative than the designers' estimation.
    In particular, an ideal adversary with $\omega = 2$ could break up to $r_\mathcal{H} = 29$ rounds below $128$ bits.
    Moreover, if the $\Fq$-Boolean Macaulay matrix can be constructed with negligible cost, then an ideal adversary could break up to $r_\mathcal{H} = 34$ rounds.

	\begin{table}[H]
		\centering
		\caption{\Hydra complexity estimations for recomputing the DRL Gr\"obner basis.
			The column $d_{\reg}$ shows the degree of regularity under the semi-regularity assumption.
			All estimations use $\omega = 2$.}
		\label{Tab: Hydra recomputing DRL Groebner Basis}
		\resizebox{\linewidth}{!}{
			\begin{tabular}{ c | c | M{20mm} | M{19mm} | M{20mm} | M{19mm} || M{21mm} }
				\toprule
				\multicolumn{2}{ c | }{} & \multicolumn{5}{ c }{Complexity (bits)} \\
				\midrule
				\multicolumn{2}{ c | }{} & \multicolumn{2}{ c | }{Proven Solving Degree} & \multicolumn{2}{ c ||  }{Semi-Regular} & \\
				\midrule
				$r_{\mathcal{H}}$ & $d_{\reg}$ & Boolean Matrix Construction & Boolean Gaussian Elimination &  Boolean Matrix Construction & Boolean Gaussian Elimination &Semi-Regular Estimate \cite{EC:GOSW23} \\
				\midrule

				$28$ & $22$ & $148.72$ & $110$ & $118.28$ & $102.10$ & $125.31$ \\
                $29$ & $23$ & $153.89$ & $114$ & $123.29$ & $106.25$ & $130.80$ \\
                $30$ & $24$ & $159.05$ & $118$ & $128.27$ & $110.40$ & $136.29$ \\
                $31$ & $25$ & $164.21$ & $122$ & $133.25$ & $114.54$ & $141.77$ \\
                $32$ & $26$ & $169.37$ & $126$ & $138.21$ & $118.66$ & $147.26$ \\
                $33$ & $27$ & $174.52$ & $130$ & $143.16$ & $122.79$ & $152.75$ \\
                $34$ & $28$ & $179.67$ & $134$ & $148.10$ & $126.90$ & $158.23$ \\
                $35$ & $28$ & $184.82$ & $138$ & $150.37$ & $129.66$ & $160.24$ \\
                $39$ & $32$ & $205.41$ & $154$ & $170.12$ & $146.16$ & $182.22$ \\
                $45$ & $37$ & $236.24$ & $178$ & $197.03$ & $169.55$ & $211.72$ \\

				\bottomrule
			\end{tabular}
		}
	\end{table}

    \paragraph{Empirical F4 Working Degree.}
    The \Hydra designers performed small scale Gr\"obner basis experiments for \Hydra with the F4 \cite{Faugere-F4} algorithm, see \cite[Table~4]{EPRINT:GOSW22}.
    With our \OSCAR implementation we replicated these experiments to verify that F4's highest working degree matches the semi-regular degree of regularity before and after the change of coordinates.
    Our results are recorded in \Cref{Tab: Hydra F4 working degree}.
    Our experiments have been performed on an AMD EPYC-Rome (48) CPU with 94 GB RAM.

    In all our experiments F4's highest working degree indeed matches the semi-regular degree of regularity.
    As expected, after the change of coordinates the computation is faster, because F4 will implicitly eliminate all monomials which are divisible by a square.

    \clearpage 
    \begin{table}[H]
        \centering
        \caption{\Hydra highest working degree and total running time of F4 over the prime field $\F_{7741}$.
            The column $d_{\reg}$ shows the degree of regularity under the semi-regularity assumption.}
        \label{Tab: Hydra F4 working degree}
        \resizebox{\linewidth}{!}{
            \begin{tabular}{ c | c | M{20mm} | M{14mm} || M{20mm} | M{14mm} || M{20mm} | M{16mm} }
                \toprule
                \multicolumn{2}{ c | }{} & \multicolumn{ 2 }{ M{34mm} || }{Before Change of Coordinates} & \multicolumn{2}{ M{34mm} || }{After Change of Coordinates} & \multicolumn{2}{ M{36mm} }{Grassi et al.'s data \cite[Table 4]{EPRINT:GOSW22}} \\
                \midrule
                $r_\mathcal{H}$ & $d_{\reg}$ & Highest F4 working degree & Total time (s) & Highest F4 working degree & Total time (s) & Highest F4 working degree & Total time (s) \\
                \midrule

                $3$ & $3$ & $3$ & $0.01$ & $3$ & $0.01$ & - & - \\
                $4$ & $3$ & $3$ & $0.01$ & $3$ & $0.01$ & - & - \\
                $5$ & $4$ & $4$ & $0.02$ & $4$ & $0.01$ & - & - \\
                $6$ & $5$ & $5$ & $0.16$ & $5$ & $0.03$ & $5$ & $0.1$ \\
                $7$ & $5$ & $5$ & $1.35$ & $5$ & $0.25$ & $5$ & $0.9$ \\
                $8$ & $6$ & $6$ & $14.88$ & $6$ & $2.68$ & $6$ & $13.4$ \\
                $9$ & $7$ & $7$ & $384.13$ & $7$ & $67.19$ & $7$ & $293.8$ \\
                $10$ & $8$ & $8$ & $8676.30$ & $8$ & $3194.95$ & $8$ & $8799.5$ \\
                $11$ & $8$ & - & - & - & - & $8$ & $229606.8$ \\

                \bottomrule
            \end{tabular}
        }
    \end{table}

    \subsubsection{Term Order Conversion \& Eigenvalue Method}\label{Sec: term order conversion}
    With a \Hydra DRL Gr\"obner basis at hand, see \Cref{Sec: Extracting a Groebner Basis}, we can proceed to direct system solving techniques.
    For example, we can use a generic algorithm like the state-of-the-art probabilistic FGLM algorithm \cite{Faugere-SubCubic} to convert to a LEX Gr\"obner basis, and then factor the univariate polynomial.
    Assuming that the \Hydra DRL Gr\"obner basis is in shape position the complexity of term order conversion is bounded by \cite[Proposition~3]{Faugere-SubCubic}
    \begin{equation}\label{Equ: Hydra FGLM complexity}
        \mathcal{O} \Big( (2 \cdot r_\mathcal{H} - 2) \cdot 2^{2 \cdot r_\mathcal{H} - 2} \cdot \big( 2^{(\omega - 1) \cdot (2 \cdot r_\mathcal{H} - 2)} + (2 \cdot r_\mathcal{H} - 2)^2 \cdot \log_2 (2 \cdot r_\mathcal{H} - 2)  \big) \Big)
    \end{equation}
    field operations.
    Next we have to factor the univariate polynomial of degree $2^{2 \cdot r_\mathcal{H} - 2}$ via the GCD with the field equation.
    However, even for full rounds $r_\mathcal{H} = 39$ this degree is well below $2^{128}$.
    Therefore, the GCD complexity from \Cref{Equ: GCD with field equation complexity} will not exceed $128$ bits, and henceforth can be ignored.

    On the other hand, the \Hydra DRL Gr\"obner basis satisfies the structure assumption from \Cref{Equ: special shape}, so we can pass to the dedicated Eigenvalue Method to find a $\Fq$-valued solution.
    For \Hydra, the complexity from \Cref{Equ: complexity solving structured systems} then evaluates to
    \begin{equation}\label{Equ: Hydra: Eigenvalue Methid complexity}
        \mathcal{O} \left( 2^{\omega + 1} \cdot \frac{2^{\omega \cdot (2 \cdot r_\mathcal{H} - 2)} - N^{r_\mathcal{H} - 1}}{2^{2 \cdot \omega} - N} \right),
    \end{equation}
    where $N$ is a bound on the number of $\Fq$-valued solutions in every iteration.
    Again, the GCD complexity can be considered negligible.
    Since a bound on $N$ is in general not available, we consider a worst-case scenario approach.
    We assume that the first root which an adversary recovers in every iteration leads to the true \Hydra key, i.e.\ we use $N = 1$ in \Cref{Equ: Hydra: Eigenvalue Methid complexity}.

    In \Cref{Tab: Hydra FGLM Eigenvalue Method} we evaluated the complexities for direct system solving techniques and compared them to the other strategies.
    Our estimations are again from a designer's point of view which assumes an ideal adversary who can achieve the linear algebra constant $\omega = 2$.
    In this scenario, term order conversion breaks up to $r_\mathcal{H} = 31$ rounds below $128$ bits.
    This invalidates the claim of the \Hydra designers \cite[\S 7.4]{EC:GOSW23} that $r_\mathcal{H} = 31$ rounds are sufficient to provide $128$ bits of security against Gr\"obner basis attacks for an ideal adversary.
    Moreover, via the dedicated Eigenvalue Method up to $r_\mathcal{H} = 33$ do not achieve $128$ bits of security.
    For full rounds $r_\mathcal{H} = 39$ the security margin is reduced by $\approx 30$ bits, however the security claim is not affected.

    Let $r_\mathcal{H}^\ast$ be the minimum round number which achieves $128$ bits of security in a Gr\"obner basis attack, i.e.\ according to our analysis we have $r_\mathcal{H}^\ast = 34$.
    Recall from \Cref{Equ: head round numbers} that the head round number is set as $r_\mathcal{H} = \ceil{1.25 \cdot \max \left\{ 24, 2 + r_\mathcal{H}^\ast \right\}}$.
    Thus, our analysis suggests that this round number is increased to $r_\mathcal{H} = 45$ to recover the security margin originally intended by the designers.

	In case one wants to be extra conservative, one can consider the construction of the $\Fq$-Boolean Macaulay matrix to be for free.
	Then, we have $r_\mathcal{H}^\ast = 35$, see \Cref{Tab: Hydra recomputing DRL Groebner Basis}, which implies $r_\mathcal{H} = 47$.

    \begin{table}[H]
        \centering
        \caption{\Hydra complexity estimations for term order conversion and the Eigenvalue Method.
                 The column $d_{\reg}$ shows the degree of regularity under the semi-regularity assumption.
                 All estimations use $\omega = 2$.}
        \label{Tab: Hydra FGLM Eigenvalue Method}
        \begin{tabular}{ c | c | M{18mm} | M{16mm} | M{21mm} || M{21mm} }
            \toprule
            \multicolumn{2}{ c | }{} & \multicolumn{4}{ c }{Complexity (Bits)} \\
            \midrule
            $r_\mathcal{H}$ & $d_{\reg}$ & Term Order Conversion & Eigenvalue Method & Boolean Semi-Regular Estimate & Semi-Regular Estimate \cite{EC:GOSW23} \\
            \midrule

            $28$ & $22$ & $113.75$ & $107.09$ & $118.28$ & $125.31$ \\
            $29$ & $23$ & $117.81$ & $111.09$ & $123.29$ & $130.80$ \\
            $30$ & $24$ & $121.86$ & $115.09$ & $128.27$ & $136.29$ \\
            $31$ & $25$ & $125.91$ & $119.09$ & $133.25$ & $141.77$ \\
            $32$ & $26$ & $129.95$ & $123.09$ & $138.21$ & $147.26$ \\
            $33$ & $27$ & $134.00$ & $127.09$ & $143.16$ & $152.75$ \\
            $34$ & $28$ & $138.04$ & $131.09$ & $148.10$ & $158.23$ \\
            $35$ & $28$ & $142.09$ & $135.09$ & $150.37$ & $160.24$ \\
            $39$ & $32$ & $158.25$ & $151.09$ & $170.12$ & $182.22$ \\
            $45$ & $37$ & $182.46$ & $175.09$ & $197.03$ & $211.72$ \\

            \bottomrule
        \end{tabular}
    \end{table}

    \paragraph*{Small Scale Experiments.}
    With our \OSCAR implementation we computed the characteristic polynomial of the \Hydra multiplication matrix $\mathbf{M}_{\hat{x}_{2 \cdot r_\mathcal{H} - 2}}$, see \Cref{Equ: characteristic polynomial},  for small $r_\mathcal{H}$ and measured the time of the overall computation.
    In addition, we computed the polynomial matrix determinant from \Cref{Equ: characteristic polynomial} for uniformly random matrices $\mathbf{A}_0, \mathbf{A}_1 \in \Fq^{(2 \cdot r_\mathcal{H} - 3) \times (2 \cdot r_\mathcal{H} - 3)}$.
    For construction and determinant of the matrix we used \OSCAR's functions for multivariate polynomial division and determinant respectively.
    Our results are recorded in \Cref{Tab: Hydra characteristic polynomial}.
    In addition, we performed term order conversion to LEX for our \Hydra DRL Gr\"obner basis with \OSCAR's \texttt{fglm} function.
    For comparison, we also generated a random multivariate quadratic polynomial system in the same number of variables as \Hydra, computed the DRL Gr\"obner basis, and finally converted to LEX via FGLM.
    Our FGLM results are recorded in \Cref{Tab: Hydra FGLM}.
    All experiments in this section have been performed on an AMD EPYC-Rome (48) CPU with 94 GB RAM.

    For the characteristic polynomial, we observed that the construction time for the multiplication matrix is negligible compared to the one for the characteristic polynomial.
    Moreover, \Hydra characteristic polynomials are computed faster than for random matrices $\mathbf{A}_0$ and $\mathbf{A}_1$.

    For term order conversion, we observed that FGLM on \Hydra systems requires less time than for random multivariate quadratic systems.
    We also note that in all our experiments the \Hydra LEX Gr\"obner basis was in $\hat{x}_{2 \cdot r_\mathcal{H} - 2}$-shape position.

    \clearpage 
    \begin{table}[H]
        \centering
        \caption{\Hydra running times for construction of the multiplication matrix and its characteristic polynomial over the prime field $\F_{7741}$.
            The column ``Determinant Random Matrix'' shows the running time of the determinant from \Cref{Equ: characteristic polynomial} with uniformly random $\mathbf{A}_0, \mathbf{A}_1 \in\F_{7741}^{(2 \cdot r_\mathcal{H} - 3) \times (2 \cdot r_\mathcal{H} - 3)}$.
        }
        \label{Tab: Hydra characteristic polynomial}
        \begin{tabular}{ c | M{21mm} | M{21mm} || M{25mm} }
            \toprule
            & \multicolumn{3}{ c }{Time (s)} \\
            \midrule
            $r_\mathcal{H}$ & Construction \Hydra Matrix & Characteristic Polynomial \Hydra Matrix & Determinant Random Matrix \\
            \midrule

            $5$ & $< 0.1$ & $0.3$    & $0.4$     \\
            $6$ & $0.7$   & $73.9$   & $82.8$    \\
            $7$ & $11.1$  & $8336.7$ & $11722.9$ \\

            \bottomrule
        \end{tabular}
    \end{table}

    \begin{table}[H]
        \centering
        \caption{\Hydra FGLM running times over the prime field $\F_{7741}$.
                 The column ``Random Quadratic F4'' shows the running time for computing a DRL Gr\"obner basis with F4 for a random multivariate quadratic polynomial system in the same number of variables as the \Hydra system.
                 The column ``Random Quadratic FGLM'' shows the running time of FGLM for the random DRL Gr\"obner basis.
            }
        \label{Tab: Hydra FGLM}
        \begin{tabular}{c | c || c | c }
            \toprule
            & \multicolumn{3}{ c }{Time (s)} \\
            \midrule
            $r_\mathcal{H}$ & \Hydra FGLM & Random Quadratic F4 & Random Quadratic FGLM \\
            \midrule

            $5$ & $0.3$     & $0.1$  & $0.4$    \\
            $6$ & $20.0$    & $0.8$  & $26.1$   \\
            $7$ & $1265.7$  & $15.6$ & $1598.5$ \\
            $8$ & $78345.1$ & $1123.9$ & $99492.9$ \\

            \bottomrule
        \end{tabular}
    \end{table}

    \section*{Acknowledgments}
    Matthias Steiner has been supported in part by the enCRYPTON project (grant agreement No.\ 101079319).

	\pdfbookmark[1]{References}{References}
    \bibliographystyle{alphaurl}
    \bibliography{biblio.bib,abbrev3.bib,crypto.bib}

\end{document}